\documentclass[a4paper,11pt]{article}

\usepackage{fullpage}
\usepackage{xcolor}
\usepackage{amssymb}
\usepackage{amsmath}
\usepackage{subcaption}
\usepackage{framed}
\usepackage{footmisc}
\usepackage[round,authoryear]{natbib}
\usepackage{authblk}
\usepackage{colortbl}
\definecolor{webgreen}{rgb}{0,0.4,0}
\definecolor{webbrown}{rgb}{0.6,0,0}
\definecolor{purple}{rgb}{0.5,0,0.25}
\definecolor{darkblue}{rgb}{0,0,0.7}
\definecolor{darkred}{rgb}{0.7,0,0}
\usepackage{hyperref}
\hypersetup{colorlinks,citecolor=darkblue,filecolor=black,linkcolor=darkred,urlcolor=webgreen,pdfpagemode=UseNone,
pdfstartview=FitH}
\newcommand{\ignore}[1]{}

\usepackage{amsthm}
\newtheorem{lemma}{{\sc Lemma}}

\newtheorem{corollary}{{\sc Corollary}}
\newtheorem{theorem}{{\sc Theorem}}
\newtheorem{definition}{{\sc Definition}}

\newtheorem{claim}{{\sc Claim}}

\usepackage{cleveref}
\crefname{claim}{claim}{claims}
\crefname{fact}{fact}{facts}
\crefname{algorithm}{algorithm}{algorithms}
\crefname{observation}{observation}{observations}
\crefname{equation}{equation}{equations}
\crefname{assumption}{assumption}{assumptions}

\DeclareMathOperator*{\argmax}{\arg\!\max}

\usepackage{pgf}
\usepackage{verbatim}
\usepackage{enumerate}
\usepackage{amssymb}
\usepackage{mathrsfs}
\usepackage{algorithm}
\usepackage[noend]{algorithmic}
\usepackage{resizegather}
\usepackage{url}

\usepackage{xspace}

\usepackage{nicefrac}

\usepackage{enumitem}
\usepackage{multirow}
\usepackage{dsfont}

\newcommand{\mechabbrv}{\text{\tt ECATS}} 
\newcommand{\mech}{{\bf \texttt{E}}galitarian and {\bf \texttt{C}}ongestion {\bf 
\texttt{A}}ware {\bf \texttt{T}}ruthful {\bf \texttt{S}}lot allocation}
\usepackage{cleveref}
\crefname{claim}{claim}{claims}
\crefname{fact}{fact}{facts}
\crefname{algorithm}{algorithm}{algorithms}
\crefname{observation}{observation}{observations}
\crefname{equation}{equation}{equations}
\crefname{assumption}{assumption}{assumptions}
\crefname{hypothesis}{hypothesis}{hypotheses}
\usepackage{enumerate}
\usepackage{subcaption}
\usepackage[noend]{algorithmic}
\usepackage{resizegather}
\usepackage{enumitem}
\usepackage{multirow,bigdelim}
\usepackage{cancel}
\usepackage{framed}
\usepackage{wrapfig}
\usepackage{tfrupee}
\usepackage{multirow}
\usepackage{footmisc}

\renewenvironment{itemize}{
  \begin{list}{\textbullet}{
    \setlength{\leftmargin}{1.5em}
    \setlength{\itemsep}{0.25em}
    \setlength{\parskip}{0pt}
    \setlength{\parsep}{0.25em}
  }
}{
  \end{list}
}

\title{\bf Egalitarian and Congestion Aware Truthful Airport Slot Allocation Mechanism}

\author[1]{Aasheesh Dixit$^*$}
\author[2]{Garima Shakya$^*$}
\author[1]{Suresh Kumar Jakhar}
\author[2]{Swaprava Nath}

\affil[1]{\small Indian Institute of Management Lucknow, \texttt{\{fpm18021, skj\}@iiml.ac.in}}
\affil[2]{\small Indian Institute of Technology Kanpur, \texttt{\{garima, swaprava\}@cse.iitk.ac.in}, $^*$: Equal contribution}

\date{}

\begin{document}
\maketitle

\begin{abstract}
\noindent
We propose a mechanism to allocate slots fairly at congested airports. This mechanism: (a) ensures that the slots are allocated according to the {\em true} valuations of airlines, (b) provides fair opportunities to the flights connecting remote cities to large airports, and (c) controls the number of flights in each slot to minimize congestion. The mechanism draws inspiration from economic theory. It allocates the slots based on an {\em affine maximizer} allocation rule and charges payments to the airlines such that they are incentivized to reveal their true valuations. The allocation also optimizes the occupancy of every slot to keep them as uncongested as possible.
The formulation solves an optimal integral solution in strongly polynomial time. We conduct experiments on the data collected from two major airports in India. We also compare our results with existing allocations and also with the allocations based on the International Air Transport Association (IATA) guidelines. The computational results show that the {\em social utility} generated using our mechanism is 20-30\% higher than IATA and current allocations. 
\end{abstract}

\smallskip \noindent
{\bf Keywords:} airport slot allocation, congestion cost, social welfare, mechanism design, strongly polynomial algorithm.

% Paper body
\section{Introduction}
In the last two decades or so, an overwhelming increase in demand for air transportation coupled with political, physical, and institutional constraints for capacity expansion resulted in the congestion and delays at the world's major commercial airports. As demand for airport slots \footnote{The daily runway scheduling period of an airport is divided into time intervals of a fixed length (e.g., 15 minutes) called slots \citep{androutsopoulos2020modeling}. The slot allocation to an airline movement means the permission to use the airport for landing or taking-off within a particular time interval (slot).  } exceeds its available capacity, it results in more holding time for the permission to land or take-off, which leads to congestion. Airport congestion is veritably imposing a tremendous cost on the world economy which includes additional aircraft operating costs, passenger delay costs, etc. The other externalities are environmental and noise pollution around the congested airport, while aircraft wait in queue with engines fired up. A study commissioned by the Federal Aviation Administration (FAA) estimates the total cost of delay to-be around \$31.2 billion for the calendar year 2007 \citep{deshpande2012impact}. The contemporary research on airport slot allocations is primarily focused on demand-side solutions to mitigate congestion as it has the potential to restore the demand-capacity balance over a medium to short time horizon with fairly low investments \citep{barnhart2004airline}. The demand management strategies to manage slots at congested airport ranges from various administrative tools to market-based mechanisms.

The administrative approaches for slot allocation include grandfather rights, first-come-first-serve, lotteries, new entrant rule, etc.  \citep{ball2006auctions}. These approaches essentially are government or institutional interventions through rules and regulations to allocate scarce slots at congested airports. In the grandfathering rights, the historical precedence of the allocated slots is maintained \citep{sieg2010grandfather}. Whereas, in the first-come-first-serve rule, the slots are assigned based on their time of arrival, and airlines queue up for runway and gate access \citep{fan2002practical}. In the new entrant rule, the preference is given to new entrants for the vacant/newly available slots. For example, a considerable number of slots at the newly constructed runway at Frankfurt airport were allocated to new entrants \citep{IATA}. The use of administrative instruments may also be warranted to achieve certain other social goals. For example, governments across the world have used a varied set of administrative instruments to promote remote connectivity. Indeed, air transportation has emerged as a `lifeline service' to connect remote cities where land transportation is not a real option \citep{fageda2018air}. The link between air connectivity and economic growth is also well established \citep{brueckner2003airline,alderighi2017hidden}. The study by \citet{bilotkach2015airline} shows a strong correlation between regional growth and connectivity to various new destinations. \citet{fageda2018air} observe that an appropriate mechanism can help governments effectively promote regional air transport, which otherwise is excluded under normal market conditions. An administrative approach may also respond to a country's varied political goals, for example, UK governments' objective to connect Heathrow airport to the rest of the UK \citep{burghouwt2017influencing}.

The extant literature also points out several issues with administrative approaches, with one of the biggest criticisms being that they are economically inefficient. The runway resource at congested airports in peak periods is valuable and scarce, and grandfathering or lotteries are inefficient and random ways to allocate them \citep{zografos2003critical,zografos2013decision}. Therefore, the slots should be allocated to the airlines which value them the most \citep{cao2000value}. Another criticism for grandfathering rule is that it creates entry barriers to new airlines \citep{vaze2012modeling} and encourages legacy carriers to overschedule flights to avoid losing the allocated slots \citep{harsha2009mitigating}. This also prevents the effective competitive pressure on these incumbents from the new entrants, especially the low-cost carriers. For example, British Airways holds 51\% of the take-off and landing capacity at London Heathrow (the busiest airport in Europe) with an estimated worth of 742 million pounds \citep{horton_2020}. This airport operates at its 100\% capacity, which makes it incredibly difficult for the new entrant to obtain a slot. Moreover, when there are no monetary payments involved in slot allocation, the airlines may overdemand the slots to minimize deviation from their preferred schedule  \citep{vaze2012modeling} or gain market power through slot hoarding  \citep{sheng2019modeling}. Historically it has been established that government interventions lead to inefficiencies, and airlines may exploit the loopholes in the regulations \citep{burghouwt2017influencing}.

Researchers have shown that market-based mechanisms result in the efficient allocation of scarce resources \citep{zhang2020truthful, abdulkadirouglu2003school}. \citet{czerny2014airport} concluded that the market mechanisms would essentially reveal the resources' true economic value at the congested airport. They provide a flexible and transparent approach for balancing supply-demand mismatch through pricing. These mechanisms also provide equal opportunities to the legacy carriers and the new entrant and promote healthy competition. Various market-based mechanisms such as congestion pricing, auctions, and secondary trading of slots have been proposed. Congestion pricing favors charging airlines a fee based on several proposed rules such as: 
\begin{itemize}
    \item  Pricing based on the marginal cost of delays \citep{vickrey1969congestion,carlin1970marginal}.
    \item Differential pricing across carriers by considering small and large airlines \citep{brueckner2009price}.
    \item Pricing based on contribution to the infeasibility of the ideal solution \citep{castelli2012airport}.
    \item Pricing based on types of passengers when airlines price discriminates \citep{czerny2014airport}.
    \item	Pricing based on the desired amount of delay that airlines are ready to buy \citep{mehta2020incentive}. 
\end{itemize}    
Congestion pricing results in considerable welfare gains, which can minimize total delay and attain ideal slot allocations \citep{daniel2009pricing,czerny2010airport}. \citet{daniel2011congestion} reported that congestion pricing could help in saving of \$72 million to \$105 million annually at Canadian airports by reducing delays and associated costs. A drawback of congestion pricing is that the fee has to be iteratively varied in time and among different airports depending on the degree of congestion set up by the airport administration. The other market-based approach widely discussed in the extant literature is auctioning, where airline bids for slots. An auction-based mechanism maximizes the social welfare as it allocates slots purely based on valuation maximization \citep{harsha2009mitigating}. \citet{basso2010pricing} compared congestion pricing with slot auctions and conclude that there is no clear winner for airport profit maximization. \citet{ball2020quantity} proposed a quantity-contingent-based auction mechanism to allocate the slots at the congested airports. The authors imposed a constraint on the accepted bids to control the market power of airlines. The authors proposed using Vickrey-Clarke-Groves (VCG) payment; however, they failed to prove the individual rationality and computational complexity of the problem.

Despite being regarded as the most efficient way to allocate slots at congested airports, the market-based mechanisms also have certain shortcomings. They are considered detrimental for flights from remote cities due to their low valuations compared to flights from big cities. The low profitability and inconsistent load factor of movements from remote cities may limit their ability to win slots at an auction or pay pure market-based congestion prices. Any mechanism solely based on the transfer of money will be unfavorable to remote communities and will lead to the exclusion of air-service to these cities \citep{green2007airports,harsha2009mitigating,sheard2014airports}. The use of administrative instruments may be warranted to achieve social objectives by ensuring connectivity to peripheral regions and support the population of remote regions. The need to combine market-based and administrative instruments leads to a hybrid mechanism, which allocates the slots based on efficiency goal (valuation maximization) and ensures slot opportunities for flights connecting to remote cities. This paper addresses this dual goal of valuation-based slot allocation with remote city connectivity (social objective) at the congested airport.

The key lever for congestion minimization is to limit the number of allocated movements (landings/take-offs) in a particular slot. The most important decision here is to determine the number of movements to assign based on the trade-off between the cost of delay and resource utilization. In our proposed mechanism, we determine the number of movements in each slot based on the trade-off between an increase in valuations due to additionally allocated movement and the resultant increase in congestion cost. It considers a proper balance between flight delays and the extent of the services offered (the number of flights scheduled). As the number of movements scheduled in each slot decreases, it would decrease congestion level and resultant delays. Based on a case study of LaGuardia Airport, \citet{li2010optimal} found that airport congestion can be minimized by limiting the number of allocated movements in a slot. Similarly, \citet{swaroop2012more} found that more than two-thirds of the total system-delays can be reduced by capping the slot allocation. We propose that if the number of movements in a slot is allocated up to a certain limit, it may not result in significant congestion and delays. However, if we allocate movements in a slot beyond a certain limit, it will start adding congestion and delays. Therefore, we have imposed a penalty in congestion cost for every additional allocated movement beyond a limit. The determination of allocation limit depends on various factors such as weather conditions in a particular season, aircraft mix and skill of the air traffic controllers and pilots, and other factors. Historical airport data can be used to determine these allocation limits.

\subsection*{Our Contributions}
The goal of this paper is to devise a slot allocation mechanism that can integrate various (administrative and market-based) instruments into an overall slot allocation strategy that satisfies multiple criteria such as congestion mitigation, weighted efficiency \footnote{Efficiency usually implies maximizing the sum of allocated valuations, which can be biased towards flights from metro cities. The weighted efficiency essentially means the sum of weighted valuations with larger weights for the flights from remote cities to provide a fair opportunity in the allocation}, and truthfulness (considering that the airlines' true valuation is their private information). It is worth noting that there is a dearth of models in the extant literature that can simultaneously consider the efficiency and social welfare goals in slot allocation. In this paper, we propose a mechanism with the following distinctive features:
\begin{itemize}
    \item It allocates slots based on the reported valuations of the airlines.
    \item It incorporates the remote city opportunity factor (social progress index and population of the remote city with adjustable weights), which caters to social obligations for slot allocation.
    \item It incorporates congestion cost and controls the number of movements allocated in each slot by considering the cost-benefit trade-offs. Controlling the number of allocations can be viewed as a mechanism for managing flight arrival and departure delays. The decrease in the number of allocations in each slot would decrease congestion level and resultant delays.
    \item Our mechanism provides a delicate balance between three competing goals: slot allocation based on valuations, remote city connectivity, and congestion mitigation, yet ensures computational tractability. Usually, these objectives counter are conflicting in nature.
\end{itemize}
In particular, we show the following results.
\begin{itemize}
 \item Our mechanism is truthful in dominant strategies (\Cref{thm:DSIC}). 
 \item It also incentivizes voluntary participation of the airlines (\Cref{thm:IR}). 
 \item The slot allocation problem is generally an integer program known to belong to a computationally intractable class (NP-complete). However, we show that our formulation is solvable in strongly polynomial time (\Cref{thm:Strongly_Polynomial}).
 \item The experiments (\Cref{sec:experiments}) show that the social utility generated using our mechanism is $20-30\%$ higher than IATA and Current allocations. The individual utility of movements is higher than IATA and Current allocation, while the payment made by movement is based on the type of connection provided by the flight movement.
\end{itemize}
To the best of our knowledge, this is the first mechanism that simultaneously considers multiple aspects for slot allocation, which we call \mech\ (\mechabbrv) mechanism. The remainder of the paper is organized as follows: The model and the desirable properties are formalized in \Cref{sec:model}. The mechanism and its theoretical properties are presented in \Cref{sec:mechanism,sec:TheoreticalResults} respectively. We present the experimental results in \Cref{sec:experiments}. We conclude the paper in \Cref{sec:concl}.

\section{The Model}
\label{sec:model}
 
We divide the availability of the airport in disjoint time intervals or {\em slots} \footnote{We find that the term ‘windows’ and ‘slot’ are used interchangeably in the literature. Following the definition of slot provided by \citet{androutsopoulos2020modeling},  \citet{mehta2020incentive} and \citet{ribeiro2018optimization}, we refer to each disjoint time interval as a slot with defined capacity.}. Let the set of time slots be $S=\{1,2,\dots,n \}$. Each slot $j \in S$ has a capacity $C_{j}$, which is the maximum number of flights or movements that can be accommodated simultaneously,  $C_{j} \in \mathbb{Z}_{+}, \forall j \in S$. Let $M=\{ 1,2,\dots,m\}$ be the set of movements. Every movement has a potentially different valuation for the different slots in $S$. The valuation of movement $i \in M$ for slot $j \in S$ is denoted by $v_{ij}\in \mathbb{R}_{\geqslant 0}$. For every movement $i \in M$, $v_i$ be the vector of $i$'s valuations for every slot in $S$ and $V=[v_{ij}],{i \in M, j \in S}$.  We represent the set of all feasible allocations of the flights to the slots as $A=[x_{ij}],{i \in M, j \in S}$, where, $x_{ij} = 1$ if movement $i$ is assigned slot $j$ and is equal to $0$ otherwise. We assume that each movement can be assigned to at most one slot. 

We will use the shorthand $v_i(A)$ to denote the valuation of movement $i$ in the allocation $A$, i.e., it will be equal to $v_{ik}$, if $i$ is assigned slot $k$ in $A$, and zero if $i$ is unallocated in $A$. We assume that the valuation vector $v_i$ is private information of the agent \footnote{We use the terms agent, player, and movement interchangeably in this paper.} $i$, for all $i \in M$. 

The planner decides the allocation $A$ and charge payments $p = (p_i, i \in M)$ to each of the movements. We assume that every agent wants a more valued slot to be assigned to her and also wants to pay less. Therefore, the net payoff of an agent is assumed to follow a standard {\em quasi-linear form}~\citep{SL08}
\begin{equation}
    \label{eq:utility}
    u_i((A,p), v_i) = v_i(A) - p_i.
\end{equation}
Denote the set of all allocations by ${\cal A}$ and $p_i \in \mathbb{R}, \forall i \in M$.
The planner does not know the true valuations of the agents. Therefore he needs the agents to report their valuations and decide the allocation and the payments. This leaves the opportunity for an agent to misrepresent her true valuation, e.g., an agent can report a higher valuation to get prioritized scheduling. To distinguish, we use $v_{ij}$ for the true valuation and $v'_{ij}$ for reported valuations. We will use the shorthand $v = (v_i)_{i \in M}$ to denote the true valuation profile represented as an $m \times n$ real matrix, and $v'$ to denote the reported valuation profile. The notation $v_{-i}$ denotes the valuation profile of the agents except $i$.
The decision problem of the planner is, therefore, formally captured by the following function.
\begin{definition}[Airline Scheduling Function (ASF)]
 An {\em airline scheduling function (ASF)} is a mapping $f : \mathbb{R}^{m \times n} \to {\cal A} \times \mathbb{R}^m$ that maps the reported valuations to an allocation and payment for every agent. Hence, $f(v') = (A(v'), p(v'))$, where $A$ and $p$ are the allocation and payment functions respectively \footnote{We overload the notation $A$ and $d$ to denote both functions and values of those functions, since their use will be clear from the context}.
\end{definition}

We assume that when the number of movements in a slot exceeds a pre-defined threshold, congestion and delays start kicking in. We propose a division of the slot capacity into two parts: congestion-free and congestion-prone. If the number of allocated movements in a time slot $j$ exceeds a defined threshold, $(1-\lambda)C_j$, we say that the slot $j$  is congested. The determination of the threshold, $\lambda$, depends on various factors such as airport infrastructure and local weather conditions. Historical data can be used to determine this allocation limit. 

Therefore, the following term captures the congestion level in the slot corresponding to the allocation $A$ and is equal to the number of allocated flights in the congestion-prone capacity.
    \begin{equation}\label{eq:e_j_definition}
        e_j(A)= {\left(\sum_{i \in M} x_{ij} - C_{j}(1- \lambda)\right)}^+ \quad \forall j \in S
    \end{equation}
Our results hold for $e_j$ as any linear function of $x$. 

\subsection{Remote City Opportunity Factor}
Inadequate air connectivity to metro cities is considered a major obstacle for the local economic development of remote cities. It is shown that poor air connectivity services inhibit local employment growth by limiting the city's attractiveness for new businesses and reducing the viability of existing businesses \citep{brueckner2003airline}. In the existing literature, empirical studies from European Union, United States, and China have established a positive impact of metro airport connectivity on remote cities' regional growth \citep{Brafman2003, yao2008airport}. A study based on airports in Canada shows an increase of 1126 additional air-travel passengers can create one person-year of employment \citep{benell1993regression}. \citet{yao2008airport} showed that a $10\%$ increase in population density could lead to a $1.7\%$ increase in air passenger volume in China. Therefore, the remote cities' adequate population is also an important criterion to make air connectivity from remote cities economically viable. By considering the above two factors (economic progress and population), we propose the following remote city opportunity factor (RCOF): 
\begin{equation}
    \rho_i=\alpha_{i}\frac{\gamma_{\max} - \gamma_{i} +\delta}{\sum_{i \in M} \{\gamma_{\max} -\gamma_{i}\} +\delta} + (1-\alpha_{i}) \frac{\omega_i - \omega_{\min} +\delta}{\sum_{i \in M}\{ \omega_{i} - \omega_{\min}\} +\delta}
    \label{eq:rcof}
\end{equation}
where $\gamma$ \text is the social progress index (SPI)\footnote{(\url{https://www.socialprogress.org/}} and $\omega$ \text is the population of the city.
$\gamma_{\max}$ takes the maximum value of SPI  and $\omega_{\min}$ is the minimum population among all the cities in consideration. Here $\alpha_i$ indicates the relative weight assigned to SPI and population to calculate $\rho_i$. A small factor $\delta$ with $\lim_{\delta \to 0}$, is added to avoid a division-by-zero error.The value of $\rho_i$ ranges between $0$ to $1$.

\par The SPI is an innovative way to measure the development of a region. It is a widely used index measured by the thirty-five indicators related to basic human needs, foundations of well-being, and opportunities for the city's progress. The framework is closely coherent with all sustainable development goals (SDGs) parameters. This simple but rigorous framework makes it an invaluable proxy measure of SDG performance. It captures a wide range of measures involving social and environmental factors, thus proving to be a monitoring and guiding mechanism for national policy decisions and assisting businesses in planning corporate social responsibility activities. The SPI values for various countries and their cities are available at \url{https://www.socialprogress.org}. For India, the social progress performance data of $562$ districts is collected and maintained by the Institute for Competitiveness, India. As can be seen from the above equation, the formulation is designed to prioritize the cities with a lower value of SPI and a high population to benefit from air connectivity.

\subsection{Desirable Properties}

In this section, we formally define a few desirable properties that an ASF should satisfy. Since the mechanism can only access the movements' reported values, for a truly efficient slot allocation, it is needed that the reported valuations must be the true values. The following property ensures that every movement is incentivized to reveal the values truthfully.
\begin{definition}[Dominant Strategy Truthfulness] An ASF $f=(A(\cdot), p(\cdot))$ is {\em truthful in dominant strategies} if for every $v_i, v'_{i} \in \mathbb{R}^{n}_{\geqslant 0}, i \in M$ \begin{equation*}
    v_i(A(v_i,v_{-i}))-p_i(v_i,v_{-i}) \geqslant v_i(A(v'_{i},v_{-i}))-p_i(v'_{i},v_{-i}).
\end{equation*}
\end{definition}
The inequality above shows that if the true value of agent $i$ is $v_i$, the allocation and payment resulting from reporting it `truthfully' maximizes her payoff {\em irrespective of the reports of the other agents}.

The following property ensures that it is always weakly beneficial for every rational agent to participate in such a mechanism.
\begin{definition}[Individual Rationality]
An ASF $f=(A(\cdot), p(\cdot))$ is {\em individually rational} if for every $v$, and $i \in M$
\begin{equation*}
    v_i(A(v_i,v_{-i}))-p_i(v_i,v_{-i}) \geqslant 0.
\end{equation*}
\end{definition}

For a large airport, the number of movements and slots and capacities of slots is large; this leads to an exponential increase in the size of ${\mathcal A}$. For such a setting, the allocations and payments are desired to be computed in time polynomial in the number of movements and slots. We consider mechanisms that are {\em strongly polynomial} as defined below.
\begin{definition}[Strongly Polynomial]
The algorithm runs in strongly polynomial time if  \citep{grotschel1993complexity}
\begin{enumerate}
 \item the number of arithmetic operations(addition, subtraction, multiplication, division, and comparison) in the arithmetic model of computation is bounded by a polynomial in the number of integers in the input instance; and
 \item the space used by the algorithm is bounded by a polynomial in the input size.
\end{enumerate}
\end{definition}
In the following section, we introduce the proposed mechanism.

\section{The Proposed Mechanism}
\label{sec:mechanism}

We propose the ASF \mechabbrv\ (\mech) in the quasi-linear (QL) environment. 
The allocation function uses the airlines' reported valuations for slots and the RCOF to find a socially egalitarian allocation where each movement is weighted with their RCOFs. It also puts an additive penalty for congestion and maximizes this affine sum. The following integer linear program, therefore, computes the optimal allocation.

\begin{equation}
\begin{split}
\argmax_{x} \quad & \displaystyle\sum_{j \in S} \  \sum_{i \in M}  (\rho_i \ v_{ij}\ x_{ij})\ \  -\ \  e_j(x)\ g  \\
\text{subject to} \quad & \displaystyle\sum_{i \in M} x_{ij} \leqslant C_{j} \quad \forall j \in S,\\ & \displaystyle\sum_{j \in S} x_{ij} \leqslant 1  \quad \forall i \in M\\ 
&  x_{ij} =\{ 0,1\}  \quad \forall i \in M\  \ \forall j \in S
\end{split}
\label{eq:OPT1}
\end{equation}

The first term of the objective function multiplies each movement's valuation with their RCOFs ($\rho_i$s) to equalize the opportunities to the flights to and from every city. To cater to the congestion problem, the second term of the objective function subtracts the total congestion cost.
The first constraint in the above optimization problem is the capacity constraint of each slot. The second constraint ensures that none of the movements are assigned to more than one slot.

\paragraph{Allocation} 
The allocation function $\mathcal{A}(\cdot)$ of \mechabbrv\ computes the LP relaxation of IP~\ref{eq:OPT1} as follows.

\begin{equation}
\begin{split}
\argmax_{x} \quad & \displaystyle\sum_{j \in S} \  \sum_{i \in M}  (\rho_i \ v_{ij}\ x_{ij})\ \  -\ \  e_j(x)\ g  \\
\text{subject to} \quad & \displaystyle\sum_{i \in M} x_{ij} \leqslant C_{j} \quad \forall j \in S,\\ & \displaystyle\sum_{j \in S} x_{ij} \leqslant 1  \quad \forall i \in M\\ 
&  x_{ij} \geqslant 0  \quad \forall i \in M\  \ \forall j \in S
\end{split}
\label{eq:LP_relaxation}
\end{equation}

In \Cref{sec:TheoreticalResults}, we prove that the solution of the above LP will always be integral and therefore coincides with the solution of IP~\ref{eq:OPT1}.

\paragraph{Payment}

The payment function $p$ in \mechabbrv\ is given by,
\begin{equation}
% \resizebox{0.9\hsize}{!}{
p_i(v_i,v_{-i}) = \begin{cases}
\frac{1}{\rho_i} \Bigg( h_i(v_{-i}) - \Big(\sum\limits_{k \in M \setminus \{i\}} \rho_k\ v_k\left(\ A(v_i,v_{-i}) \ \right) - g\sum\limits_{j \in S}\ e_j(A(v_i,v_{-i}))  \Big)\Bigg) &\rho_i >0\\
0&\rho_i =0
\end{cases}
% }
\label{eq:payment}
\end{equation}
where, \begin{equation*}
    h_i(v_{-i}) =\sum_{k \in M \setminus \{i\}} \rho_k\ v_k\big(A(v_{-i}) \big) - g\sum_{j \in S}\ e_j(A(v_{-i}) ) 
\end{equation*}
The payment function for an airline $i$ is proportional to the difference between the value of the optimal objective function when $i$ is absent and present, respectively. The payment is inspired by the idea of marginal contributions in the affine maximizers \citep{roberts1979characterization}.

\section{Theoretical results}
\label{sec:TheoreticalResults}
In this section, we present the theoretical guarantees for the properties of \mechabbrv. For better readability, some of the proofs are deferred to the appendix.

Our first result shows that under \mechabbrv, none of the players can get better utility by misreporting her true information.

\begin{theorem} \label{thm:DSIC} \mechabbrv\ is dominant strategy truthful.
\end{theorem}

The above result implies that irrespective of the reported valuations of the other movements, a given movement's utility is maximized when it reports its valuations truthfully. 

Our next result shows that the movements are incentivized to participate in \mechabbrv\ voluntarily.
\begin{theorem}\label{thm:IR}
\mechabbrv\ is individually rational for every movement.
\end{theorem}

The proof shows that every movement gets non-negative utility from participating in \mechabbrv.

The following few results show that even if the allocation problem of IP~\ref{eq:OPT1} falls in a computationally intractable class, its special structure can be used to find a tractable solution.

\begin{theorem} \label{thm:IntegerSolution}
The allocation of \mechabbrv\ given by LP~\ref{eq:LP_relaxation} has an integral optimal solution and is polynomially solvable.
\end{theorem}
\begin{proof}
Using the definition of $e_j$ in \Cref{eq:e_j_definition}, LP~\ref{eq:LP_relaxation} can be written as follows
\begin{align}
\label{eq:OPT2}
\centering
\begin{split}
&{\argmax_{x, w} } \quad \displaystyle\sum_{j \in S} \  \sum_{i \in M}  (\rho_i \ v_{ij}\ x_{ij})\ \  -\ \  w_j\ g \\
&\text{subject to} \quad \displaystyle\sum_{i \in M} x_{ij} \leqslant C_{j}, \quad \forall j \in S\\
& \quad \quad \displaystyle\sum_{j \in S} x_{ij} \leqslant 1,  \quad \forall i \in M\\
& \quad \displaystyle\sum_{i \in M} x_{ij} - w_j \leqslant T_j, \quad \forall j \in S\\
& \quad \quad  x_{ij}  \geqslant 0, \quad w_j \geqslant 0, \quad \forall i \in M,\  \ \forall j \in S
\end{split}
\end{align}
where, $T_j=C_{j}(1-\lambda)$.
\begin{claim}
The coefficient matrix of optimization problem~\ref{eq:OPT2} is totally unimodular (TU).
\end{claim} 
\begin{proof}
First we linearize the variables $x$ and $w$ of the optimization problem into a single vector $\bar{x}$ as $(x_{11}, x_{12}, \ldots, x_{1n}, \ldots, x_{m1}, x_{m2}, \ldots, x_{mn}, w_{1}, w_{2}, \ldots, w_{n})^\top$.
The constraints in the \Cref{eq:OPT2} can be written as $Z\bar{x} \leqslant b$, where, $Z_{(m+2n)\times(mn+n)}$ is the coefficient matrix and $b_{(m+2n)\times 1}$ is the bound vector as shown below. 
\begin{gather*}
 \begin{bmatrix}
 I_{n\times n} & I_{n\times n} & I_{n\times n} & \dots & I_{n\times n} & 0 \\
 \mathds{1}_{1\times n} & 0 &  \dots & 0 & 0 & 0 \\
 0 & \mathds{1}_{1\times n} & \dots & 0 & 0 & 0 \\
 0 & 0 & \ddots & 0 & 0 & 0\\
  0 & 0 & 0 & 0 & \mathds{1}_{1\times n} & 0 \\
 I_{n\times n} & I_{n\times n} &  I_{n\times n} & \dots & I_{n\times n} & -I_{n\times n} \\ 
\end{bmatrix}
 \begin{bmatrix} x_{11} \\ \vdots\\ x_{1n} \\ \vdots \\ x_{m1}\\ \vdots \\ x_{mn}\\ w_{1}\\ \vdots\\w_{n}\end{bmatrix}
 \leqslant
  \begin{bmatrix}
 C_{1} \\\vdots\\ C_{n}\\ 1 \\ \vdots \\ 1\\ T_{1}\\ \vdots\\ T_{n}
  \end{bmatrix}
\end{gather*} 
The first $n$ rows of $Z$ has $m$ ($n \times n $) identity matrices followed by $\mathds(0)_{n\times n}$. Each of the next $m$ rows has exactly one ($1\times n$) vector of all $1$s staggered as shown above. The last $n$ rows are similar to the first $n$ rows with the last $n$ columns being a negative $n \times n$ identity matrix.

We use the Ghouila-Houri (GH) characterization \citep{ghouila1962caracterisation,de1981some} to prove that $Z$ is TU, which says that a matrix $Z_{p \times q}$ is TU if and only if any subset $R$ of rows, $R \subseteq \{ 1,2, \cdots, p\}$, can be partitioned into two subsets $R_1$ and $R_2$, such that, $\sum_{i \in R_{1}} z_{ij} - \sum_{i \in R_{2}} z_{ij} \in \{1,0,-1\}$, where $j= 1,2,\ldots,q$.
In our case, each column of the coefficient matrix $Z$ has at most three $1$'s or one $-1$. 
Note that, the rows can be easily partitioned into {\em three} classes: class $1$ consists of first $n$ rows, class $2$ consists of next $m$ rows, class $3$ consists of last $n$ rows.
There are three exhaustive cases for the subset $R$ of rows of $Z$. 
\begin{enumerate} 
 \item $R$ consists of rows from all three classes of rows. \label{case:1}
 \item $R$ consists of rows from any two classes of rows. \label{case:2}
 \item $R$ consists of rows from exactly one class of rows. \label{case:3}
\end{enumerate}
For case~\ref{case:1}, we find the two disjoint subsets $R_1$ and $R_2$ of $R$ in an iterative manner. Begin with the partition where $R_1$ consists of rows from the first two classes of rows of $Z$ in $R$, and $R_2$ consists of rows from last class rows of $Z$ in $R$. There can be a situation where for a certain column $j$, the sum $\sum_{i \in R_1} z_{ij} = 2$ and the $\sum_{i \in R_2} z_{ij} = 0$. This can only happen when for column $j$ and rows in $R_1$, the sum of $z_{ij}$s where $i$'s are from the rows from class $1$ is $1$, for the $i$'s in class $2$ there exists a row that has $n$ $1$s intersecting with $j$, and the sum of $z_{ij}$s where $i$'s are from the rows in $R_2$ has no $1$s in $j$. In such a case, we move that row in class $2$ from $R_1$ to $R_2$. Repeat this procedure for every such column until no such situation exists. This procedure is guaranteed to converge to a partition of $R$ such that GH conditions are met.
% \sn{this construction doesn't work, think of the rows from the first two classes both having 1's but the last class does not have any 1 for a column, the difference will be 2.} Therefore, $R_1$ has $\sum_{i \in R_1} a_{ij} \in \{0,1,2\}$, while $R_2$ has $\sum_{i \in R_2} a_{ij} \in \{-1,0,1\}$. The condition for Ghouila-Houri characterization for $A$ being TU can be broken only if, $\sum_{i \in R_1} a_{ij}=2$(or $1$) and $\sum_{i \in R_2} a_{ij}=-1$. These conditions are not possible for any column $j$ of $A$, as in every column which has one $-1$, all the other entries are $0$. 

Cases~\ref{case:2} and \ref{case:3} are straightforward. In case~\ref{case:2}, the partition $R_1$ and $R_2$ are the intersections of $R$ with the respective classes of rows. For case~\ref{case:3}, any partition of $R$ satisfies GH conditions. Hence proved.
\end{proof}

The optimization problem in \Cref{eq:OPT2} has a TU coefficient matrix, which implies LP~\ref{eq:LP_relaxation} yields an optimal solution in integers and is solvable in polynomial time.
\end{proof}
\Cref{thm:IntegerSolution} shows that the solution to the LP relaxation of the allocation problem is without loss of optimality. LPs are known to be polynomially computable. However, in general, it can be weakly polynomial, i.e., the space used by the mechanism may not be bounded by a polynomial in the input size. The forthcoming results show that the solutions of allocation and payments of \mechabbrv\ are strongly polynomial. 
To show this, we reduce (LP~\ref{eq:OPT2}) to the $b$-matching problem, which is known to be strongly polynomial \citep{anstee1987polynomial}. 

\begin{definition}[$b$-Matching Problem \citep{anstee1987polynomial}]
Consider a graph $G=(V,E)$, where $V$ is the set of nodes and $E$ is the set of edges. Each edge $e_{u,v} \in E$ between any two nodes $u,v \in V$, has a cost $c_{u,v}$. Let $b=(b_1,b_2,\dots,b_{|V|})$. A $b$-matching problem for $G$ is to find the non-negative integer edge weights $w_{u,v}$ which maximises the total cost, $\sum_{u,v \in V}c_{u,v}\ w_{u,v}$ where the sum of weight on edges incident to a node $u$ is no more than $b_u$, $\forall u \in V$.
\end{definition}

\begin{theorem}\label{thm:allocationStrongly_Polynomial}
The allocation of \mechabbrv, given by LP~\ref{eq:LP_relaxation}, is implementable in strongly polynomial time.
\end{theorem}
\begin{proof}
 Consider an edge weighted bipartite graph, $G=(P,Q,E)$, where $P=M\cup \{t_j|j\in S\}$ and $Q=S$. The edges in $E$ are in two disjoint partitions, $E_1$ and $E_2$. The first partition is $E_1= \{(e_{i,j})|i \in M $ and $j\in Q\}$. Set the cost $c_{i,j}=\rho_i v_{ij}, \forall e_{i,j} \in E_1$, where $\rho_i$ and $v_{ij}$ are the RCOF and valuation of movement $i$ for slot $j$ respectively. The other partition is $E_2=\{e_{t_j,j}| t_j\in P \setminus M \text{ and } j \in Q\}$ having the cost $c_{t_j,j}=g$, where $g$ is per unit congestion cost. Define $b$ as, $b_i=1$ for every $i \in M$, $b_j=C_{j}$ for every $j \in Q$ and for every $t_j\in P \setminus M$, $b_{t_j}=\lambda C_{j}$ \footnote{This is the size of congestion prone capacity of slot $j$}. \Cref{fig:G_figure} shows the edge-weighted bipartite graph $G=(P,Q,E)$. The red and blue colored edges denote the subsets $E_1$ and $E_2$ respectively. To make the notations simpler, we represent $e_{i,j}$ as $(i,j)$.

 \begin{figure}[t]
 \leavevmode
     \centering
     \includegraphics[scale=0.5,width=0.6\textwidth]{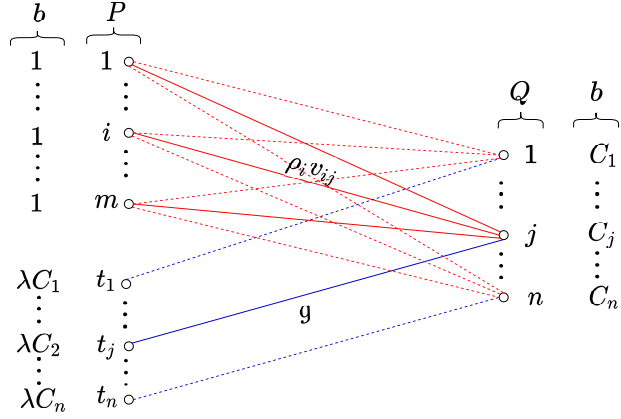}
     \caption{Constructed edge-weighted bipartite graph $G=(P,Q,E)$ and $b$}
     \label{fig:G_figure}
 \end{figure}
The objective function of the $b$-matching problem over $G$ (with $y$ being the optimization variables) is 
 \begin{equation*}
     \sum_{j \in Q} \sum_{i \in M} \rho_i v_{ij}\ y_{i,j} + \sum_{j \in Q}  g\ y_{t_j,j}
 \end{equation*}
 with the constraints, $\sum_{q \in Q} y_{p,q} \leqslant b_p$ and $ \sum_{p \in P} y_{q,p} \leqslant b_q, \ \forall p \in P,\ \forall q \in Q$.
 By the definition of $b$-matching problem, the optimal solution $y^*$ of $b$-matching of $G$ has,
 \begin{align}\label{eq:y_t_j}
    y^{*}_{t_j,j}=  \begin{cases}
C_{j}-\sum_{i\in M} y^{*}_{i,j} & \quad \quad  \sum_{i \in M} y^{*}_{i,j} - C_{j}(1- \lambda)>0\\
 \lambda C_{j} & \quad \quad \sum_{i \in M} y^{*}_{i,j} - C_{j}(1- \lambda)\leqslant 0 
\end{cases}
 \end{align}
 
The correctness of the above equation can be proved as follows. First, note that for an optimal solution $y^{*}$, if $ \exists j \in S$ with $\sum_{i \in M} y^{*}_{ij}<C_j$ then $\nexists i\in M$ such that $y^{*}_{ij} < b_i$ and $\rho_i v_{ij} \geqslant g$ -- because otherwise $y^{*}_{ij}$ can be increased maintaining the feasibility of the constraints ($b_i=1$ and $b_j=C_j$) which will increase the value of the objective function, contradicting the optimality of $y^{*}$. 

Consider case 1 of \Cref{eq:y_t_j}, $\sum_{i \in M} y^{*}_{ij} - C_{j}(1- \lambda)>0$: $\sum_{i\in M} y^{*}_{i,j}$ is the total weight on the red edges. The remaining weight $C_{j}-\sum_{i\in M} y^{*}_{i,j} $ is less than $ \lambda C_j$ and,  as $\nexists i\in M$ such that $y^{*}_{ij}< b_i$ and $\rho_i v_{ij}\geqslant g$ therefore, adding all the remaining weight to $y^{*}_{t_j,j}$ is always (a)~within the constraints of b-matching and (b)~maximize the value of the objective function. 

Consider case 2 of \Cref{eq:y_t_j}, $\sum_{i \in M} y^{*}_{ij} - C_{j}(1- \lambda) \leq 0$: the remaining weight $C_{j}-\sum_{i\in M} y^{*}_{i,j} $ is more than $ \lambda C_j$. But, $y^{*}_{t_j,j}$ can not be greater than $\lambda C_{j}$ as $b_{t_j} = \lambda C_{j}$. The optimal solution is to add the remaining weight to $y^{*}_{t_j,j}$ within the constraint $b_{t_j}=\lambda C_j$ for $j$.

Our next result formally reduces the allocation problem of \mechabbrv\ into a $b$-matching problem. We do that via a lemma, that constructs a solution $x$ for our problem from the solution $y$ of $b$-matching of $G$ as, $x_{i,j}=y_{i,j}$, for $i \in M$ and $j \in S$. Similarly, we construct a solution $y$ for $b$-matching of $G$ from a solution $x$ for our problem as, $y_{i,j}=x_{i,j}$, for $e_{i,j} \in E_1 $ and, $y_{t_j,j}$ using \Cref{eq:y_t_j} for $e_{t_j,j} \in E_2$. 
 \begin{lemma}\label{lemma:b_matching}
 Let $y^*$ is a solution for $b$-matching for graph $G$ and ${x}^*$ is s.t. ${x}^{*}_{i,j}=y^{*}_{i,j}$, $\forall i \in M$ and $ \forall j \in Q$. Then, ${x}^{*}$ is the optimal solution for the LP in \Cref{eq:LP_relaxation} iff $y^*$ is an optimal solution for $b$-matching for graph $G$.
 \end{lemma}
\begin{proof}
 First, we prove that if $y^*$ is an optimal solution for $b$-matching for graph $G$, then ${x}^{*}$ is an optimal solution for the LP in \Cref{eq:LP_relaxation}.
 \par 
  Suppose the above lemma is not true and the optimal solution for LP in \Cref{eq:LP_relaxation} is $\bar{x}$ but not ${x}^{*}$. Then, for the value of objective function in \Cref{eq:LP_relaxation},
  \begin{equation*}
      \sum_{j \in S} \  \sum_{i \in M}  \rho_i \ v_{ij}\ \bar{x}_{i,j}\ -\ e_j(\bar{x})\ g \ > \ \sum_{j \in S} \  \sum_{i \in M}  \rho_i \ v_{ij}\ x^{*}_{i,j}\ -\  e_j(x^{*})\ g
  \end{equation*}
  using the definition of $e_j(x)$ from \Cref{eq:e_j_definition},
  \begin{align*}
   \sum_{j \in S} \  \sum_{i \in M}  \rho_i \ v_{ij}\ \bar{x}_{i,j}\ -\ g\sum_{j \in S} \max \Big(\sum_{i \in M} \bar{x}_{i,j} - C_{j}(1- \lambda), & \ 0\Big) \\ >\  \sum_{j \in S} \  \sum_{i \in M}  \rho_i \ v_{ij}\ x^{*}_{i,j}\ -\  g \sum_{j \in S} \max \Big(\sum_{i \in M} x^{*}_{i,j} - & C_{j}(1- \lambda),\ 0\Big) 
\end{align*}
We divide the set of slots $S$ in two disjoint subsets $S_1,S_2$ with respect to $\bar{x}$ such that, $S_1=\{ j| \sum_{i \in M} \bar{x}_{i,j} - C_{j}(1- \lambda)>0 \}$ and $S_2=\{ j| \sum_{i \in M} \bar{x}_{i,j} - C_{j}(1- \lambda)\leqslant0 \}$. Similarly, divide $S$ in two disjoint subsets $S_3,S_4$ with respect to $x^{*}$ such that, $S_3=\{ j| \sum_{i \in M} x^{*}_{i,j} - C_{j}(1- \lambda)>0 \}$ and $S_4=\{ j| \sum_{i \in M} x^{*}_{i,j} - C_{j}(1- \lambda)\leqslant0 \}$.
 \begin{align*}
   \sum_{j \in S} \  \sum_{i \in M}  \rho_i \ v_{ij}\ \bar{x}_{i,j}\ +\ g\sum_{j \in S_1} \Big(C_{j}-\lambda  C_{j}-\sum_{i \in M}\bar{x}_{i,j} \Big)   & \\ >\   \sum_{j \in S} \ \sum_{i \in M}  \rho_i \ v_{ij}\ x^{*}_{i,j}\ +\ & g\sum_{j \in S_3} \Big(C_{j}-\lambda C_{j}-\sum_{i \in M} x^{*}_{i,j} \Big)  
\end{align*}
Adding  $g\sum_{j\in S} \lambda C_{j}$ on both sides of the above inequality,
 \begin{align*}
   \sum_{j \in S} \  \sum_{i \in M}  \rho_i \ v_{ij}\ \bar{x}_{i,j}\ +\ g\sum_{j \in S_1} \Big(C_{j}-\sum_{i \in M}\bar{x}_{i,j} \Big)  +\ g\sum_{j \in S_2} \lambda C_{j} & \\ >\   \sum_{j \in S} \  \sum_{i \in M}  \rho_i \ v_{ij}\ x^{*}_{i,j}\ +\ g\sum_{j \in S_3} \Big(C_{j}-\sum_{i \in M} x^{*}_{i,j} \Big)  & +\ g\sum_{j \in S_4} \lambda C_{j}
\end{align*}
Let $\bar{y}$ is a solution of $b$-matching for $G$ corresponding to the solution $\bar{x}$ of LP in \Cref{eq:LP_relaxation}, then the expression at right side in above inequality is the value of objective function for b-matching problem, where $\forall j \in S_1$ the case 1 of \Cref{eq:y_t_j} is true and  $\forall j \in S_2$ case 2 of \Cref{eq:y_t_j} is true. The above inequality implies that the value of the objective function of the $b$-matching problem for solution $\bar{y}$ is more than that for $y^{*}$. Therefore, the above inequality contradicts with $y^{*}$ being the optimal solution of the $b$-matching for graph $G$, which implies that ${x}^{*}$ is the optimal solution for the LP in \Cref{eq:LP_relaxation}.
\par To prove the other direction of \Cref{lemma:b_matching}, we construct $y^{*}$ from the optimal solution $x^{*}$ for the LP in \Cref{eq:LP_relaxation} as, $y^{*}_{i,j}=x^{*}_{i,j}$, for $e_{i,j} \in E_1 $ and, $y^{*}_{t_j,j}$ using \Cref{eq:y_t_j} for $e_{t_j,j} \in E_2$. The proof follows by similar argument in the reverse order.
\end{proof}
 \Cref{lemma:b_matching} shows that for every instance of the LP in \Cref{eq:LP_relaxation}, there exists an instance of $b$-matching problem such that the optimal solution of that instance of $b$-matching problem gives the optimal solution for LP in \Cref{eq:LP_relaxation}. Therefore, LP in \Cref{eq:LP_relaxation} is solvable in strongly polynomial time as there exists a combinatorial, strongly polynomial algorithm to solve the $b$-matching problem \citep{anstee1987polynomial}. 
\end{proof}
As the computation of payment for an airline $i \in M$ requires computation of a socially optimal allocation $\forall j \in M \setminus \{ i\}$ using the LP in \Cref{eq:LP_relaxation}, we get \Cref{co:PaymentStrongly_Polynomial}.

\begin{corollary}\label{co:PaymentStrongly_Polynomial}
The computation of payments for all the airlines is implementable in strongly polynomial time. 
\end{corollary}
 Following the \Cref{thm:allocationStrongly_Polynomial} and \Cref{co:PaymentStrongly_Polynomial}, we get \Cref{thm:Strongly_Polynomial}.
\begin{theorem} \label{thm:Strongly_Polynomial} There is a combinatorial strongly polynomial time algorithm for computing the allocation and payments $f=(A,p)$ in \mechabbrv.
\end{theorem}

\section{Experimental Results}
\label{sec:experiments}
In this section, we investigate the performance of $\mechabbrv$ in real-world scenarios. While $\mechabbrv$ satisfies several desirable properties of a slot allocation mechanism, its performance with varying congestion cost, use of RCOF and its relevance are not theoretically captured. This is why an experimental study is called for.
\par For the experiment, we  obtain data from two airports in India -- Indira Gandhi International Airport (DEL) and Chennai International Airport (MAA). In the year 2018, DEL handled around 70 million passengers and was the 12$^{th}$ busiest airport in the world and the 6$^{th}$ busiest in Asia. It is designated as level 3 airport and has three near-parallel runways.  
The MAA has a handling capacity of 22.5 million passengers and is the 49$^{th}$ busiest airport in Asia. The DEL is a coordinated airport\footnote{Coordinated Airports are the ones where landing airlines have to acquire landing rights and ensure its operation during a specific time period. The slots are administered by the airport operator or by a government aviation regulator. Landing/takeoff demand at these airports exceeds its capacity.} with high congestion while MAA is a non-coordinated\footnote{In non-coordinated airport, the principles governing slot allocation are less stringent and airlines periodically submit proposed schedules to the administrating authority.} airport with comparatively low congestion. The purpose of choosing these two airports is to evaluate the performance of $\mechabbrv$ in terms of the allocation and payments under different demand/congestion profiles.

\par We divide the slot capacity of the airport into two parts by setting a threshold of $(1-\lambda)$ fraction of the slot capacity. If the number of movements is less than or equal to this threshold, we consider the slot to be congestion-free and consider it congestion prone otherwise. For the experiments, we assume $\lambda = 0.2$.

We also compare $\mechabbrv$ with (a)~the {\em Current allocation}, i.e., the Current movement to slot allocation, and (b)~the {\em IATA guidelines}, which uses the optimization approach proposed by \citet{ribeiro2018optimization} and is based on the IATA guidelines.

The model of \citet{ribeiro2018optimization} minimizes the number of slot requests rejected or displaced. It optimizes the allocation based on the slot availability and airline requests, while accounting for various priorities and requirements included in the IATA guidelines.

\subsection{Summary of the data} 
\par We collect the data of flight schedule with landing and take-off slot details, city of arrival/departure, flight number and service provider. We obtain the flight movement data between January 21 to 25, 2020, for which an average of 867 movements occurred per day for DEL and an average of 220 movements happened per day for MAA airport.\footnote{We chose these dates to obtain the normal air movement patterns before COVID-19 and also because the data were most detailed during this period.} We assign the valuation to each movement based on the revenue generated by the movement. Revenue is obtained by multiplying the ticket prices with aircraft capacity and an average load factor of that origin-destination pair. The ticket price data were collected from a booking website\footnote{\url{www.ixigo.com}} for a time period of those five consecutive days for all the movements. The average load factor is defined as the ratio of the number of passengers carried and the available seats. For the data of aircraft capacity and the load factor of different destinations, we refer to the annual financial results of two public airlines: Indigo and Spicejet, which together have the largest fleet size and operates on the maximum number of routes with almost 65\% market share in India. The valuation of the movements for other time slots was calculated by generating a random value from the histogram of the valuations (which acts as an empirical distribution) for each slot. The method helps in capturing the ticket fare variation across different time slots. 
Since the prices are all in Indian Rupees (INR), the unit of the valuations is INR as well.

Next, we calculate the remote city opportunity factor (RCOF) for which we require two types of data: 
\begin{enumerate}
    \item Social progress index of the origin/destination cities.
    \item The population of the cities.
\end{enumerate} 
The social progress index data is provided by the Institute for Competitiveness, India\footnote{\url{https://competitiveness.in}}. The population data is obtained from the Office of the Registrar General and Census Commissioner, India\footnote{\url{https://censusindia.gov.in}}. The capacity of both the airports for a one-hour time slot is assumed to be the {\em average} of the maximum number of movements handled in those slots in the past.\footnote{A natural question is why we do not consider the maximum of all such number of movements. This is because in such a case, we will never witness congestion where demand is above the slot capacity with the same data.}

\begin{figure}[t!]
    \centering
    \includegraphics[width=\textwidth]{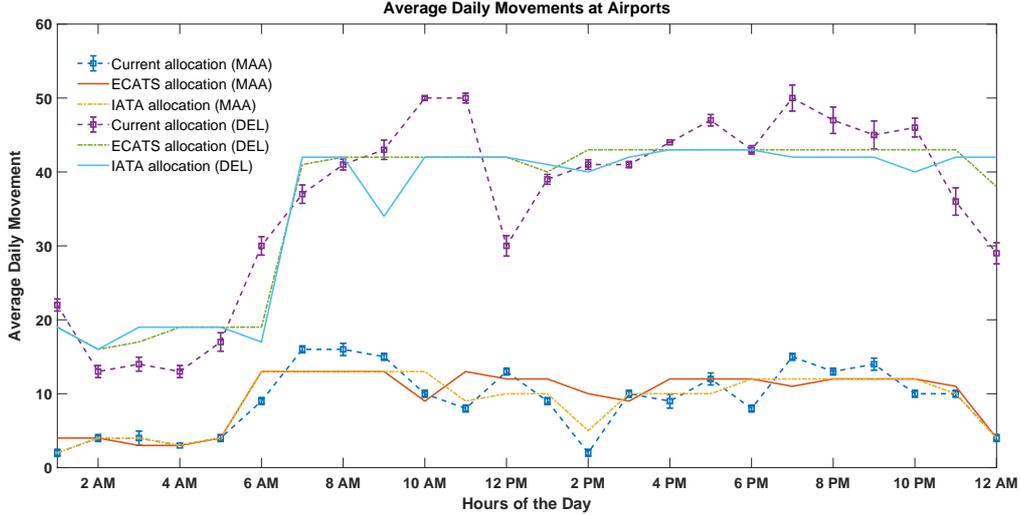}
   
    \caption{Average daily movements in the two airports.}
    \label{fig:f1}
\end{figure}
On plotting the average daily movements of the two airports (see \Cref{fig:f1}), we observe that flight movements are almost four times higher in DEL than MAA. The flight schedule has a higher variation under the Current allocation with few time slots having a large number of allocated movements. $\mechabbrv$ provides a comparatively uniform flight movement with fewer fluctuations. A stable air traffic movement minimizes losses in operational quality and improves resource allocation in the airport. Airports witness situation when terminals are `clogged' because of mis-allocation of airport resources. In other words, avoiding traffic fluctuations and having a stable traffic movements will aid in optimal utilization of airport resources. Incorporation of congestion cost plays an important role in evenness of slot allocations, which is discussed in the next subsection.

\subsection{Effect of congestion cost on the individual and social utility}
 
 We study the effect of congestion cost at the individual and social levels.
 We define the individual utility as the sum of utilities (i.e., the difference of valuation and payment for each movement) of each movement divided by the total number of allotted movements. The social utility is the value of objective function in the optimization problem given by \Cref{eq:OPT1}. Note that since the valuations are in INR, the unit of the utilities is also INR. We find that the individual utility decreases with congestion cost, but the utility is still significantly higher than the other two mechanisms (\Cref{fig:f5,fig:f6}). In addition, the utility of each movement is positive even with increasing congestion cost, which signifies that the airlines stand to gain by participating in the proposed mechanism. 

 \begin{figure}[t!]
    \centering
    \includegraphics[width=1\textwidth]{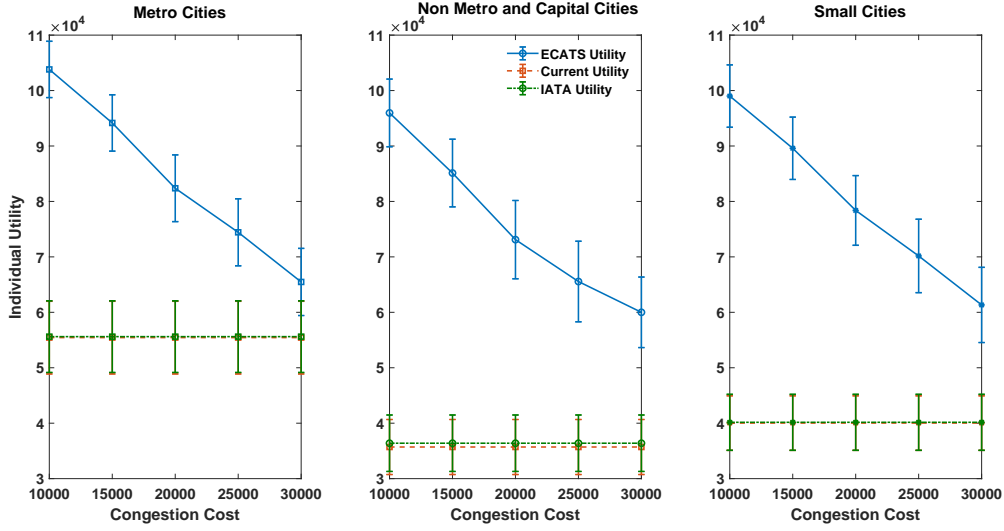}
    \caption{Individual utility of the movements at Delhi (DEL) airport.}
    \label{fig:f5}
\end{figure}

\begin{figure}[t!]
    \centering
    \includegraphics[width=1\textwidth]{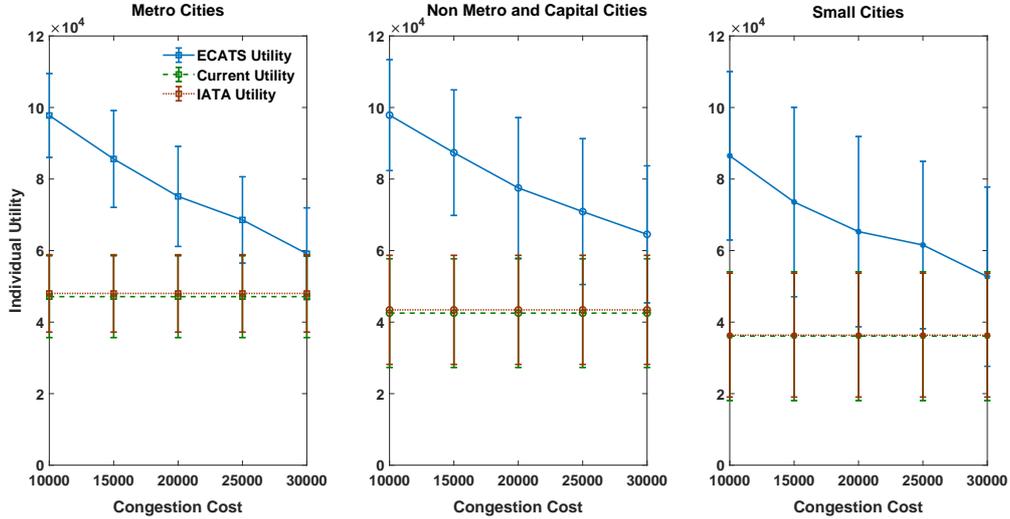}
      \caption{Individual utility of the movements at Chennai (MAA) airport.}
    \label{fig:f6}
\end{figure}

 We also observe that the individual utility is comparable for the case of different type of connections (metro, capital and remote cities). The result suggest that airlines will gain similar utility by operating in trunk routes or by providing connection to remote cities. The individual utility for $\mechabbrv$ decreases but is unchanged for the other two mechanisms. This is because there was no direct consideration of valuation maximization and congestion cost for slot allocations in the IATA and Current allocations. 
 \begin{figure}[t!]
    \centering
    \includegraphics[scale=0.3,width=1\textwidth]{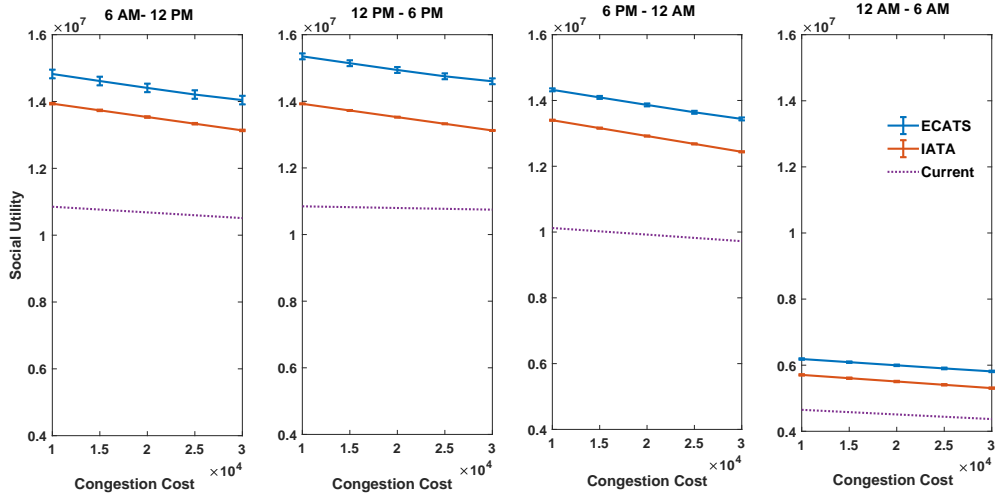}
      \caption{Social utility versus congestion cost: Delhi airport.}
    \label{fig:f7}
\end{figure}
\begin{figure}[t!]
    \centering
    \includegraphics[scale=0.3,width=1\textwidth]{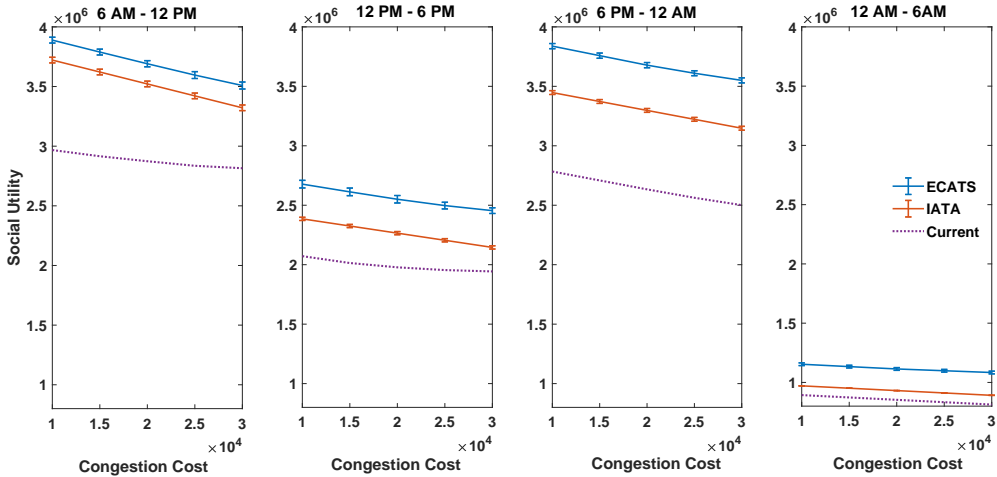}
      \caption{Social utility versus congestion cost: Chennai airport.}
    \label{fig:f8}
\end{figure}

\Cref{fig:f7,fig:f8} show that the social utility decreases with increase in congestion cost. 

As the congestion cost increases, \mechabbrv\ considers the trade-off between allocating slots above the threshold capacity (and hence increasing congestion) and rejecting the slot request. 
Only movements with high valuation are allocated these slots, and the movements with low values are rejected. We see that the social utility of $\mechabbrv$ is higher than IATA based approach as the latter mechanism is focused on minimizing the displacement of the requested slot, and does not take congestion cost into account at the time of allocation. 
 These experiments demonstrate that a value-sensitive allocation like \mechabbrv\ can improve both the individual and social utility by a significant amount.
 
\subsection{Effect of the congestion cost on the payments}
\par 
\begin{figure}[t!]
    \centering
    \includegraphics[width=1\textwidth, height=.22\paperheight]{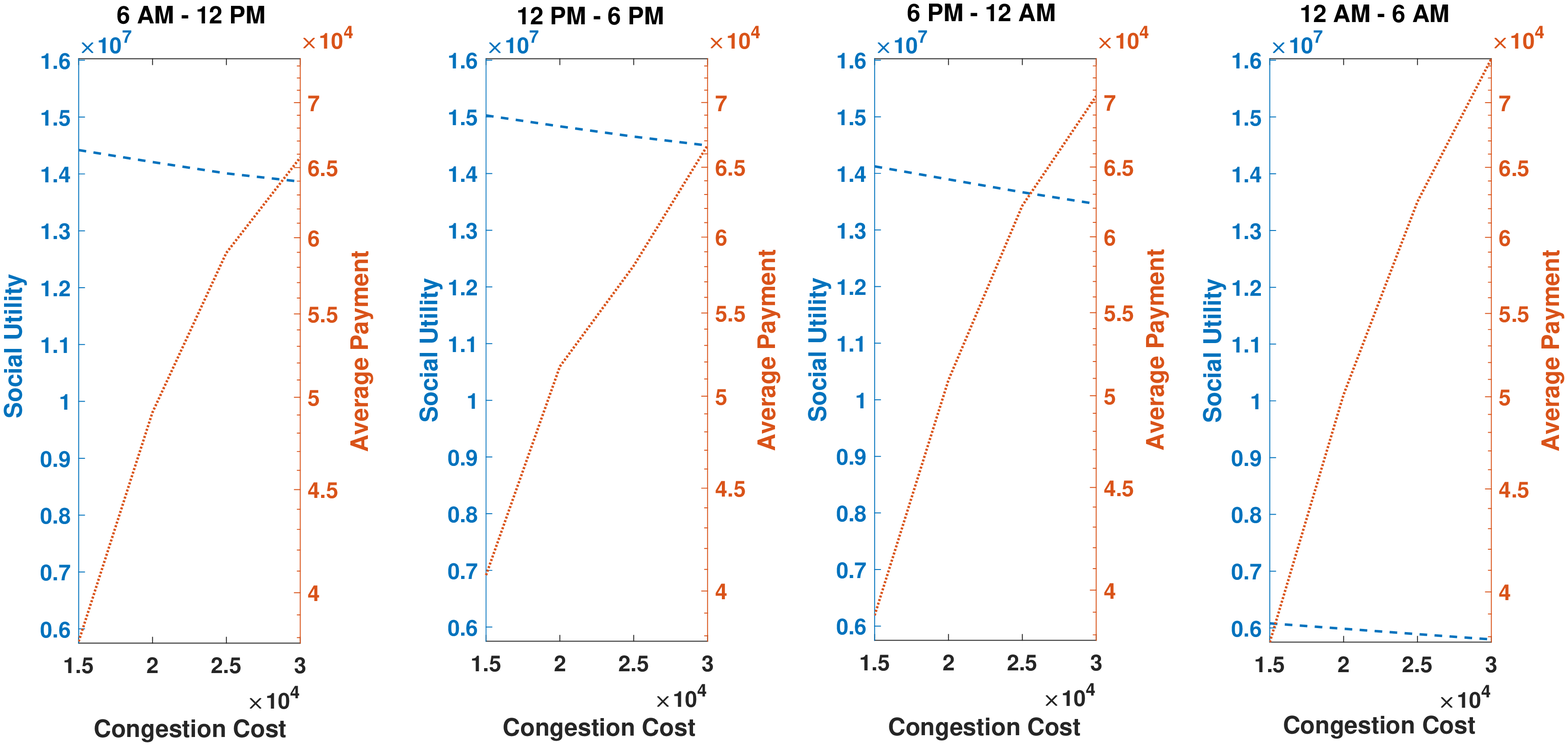}
      \caption{Social utility and average payment versus congestion cost: Delhi airport.}
    \label{fig:f3}
\end{figure}
\begin{figure}[t!]
    \centering
    \includegraphics[width=1\textwidth,height=.22\paperheight]{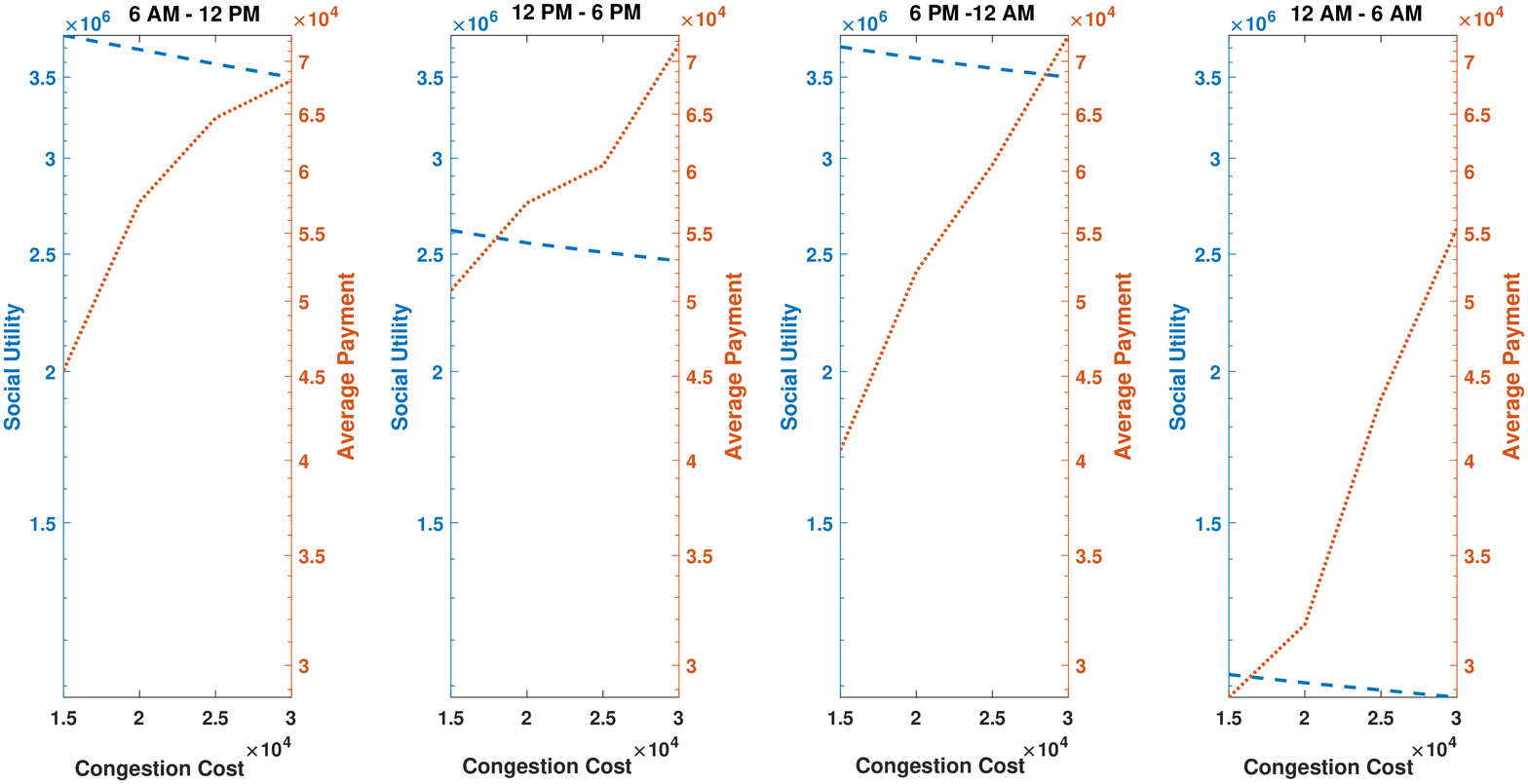}
      \caption{Social utility and average payment versus congestion cost: Chennai airport.}
    \label{fig:f4}
\end{figure}
 One of the main features of $\mechabbrv$ is the payment function that makes the mechanism truthful. In this section, we study the effect of congestion cost on resulting average payments of the movements that are allocated a slot. The payment for each movement is calculated by equation \ref{eq:payment}.
 \Cref{fig:f3,fig:f4} show that the trend of the average payments is upward, which is expected because a higher congestion cost introduces more competition among the movements for the slots. The social utility is also plotted in these figures for a reference of the scale of the payments. \Cref{fig:f3} and \ref{fig:f4} helps in depicting the increasing payment with congestion cost along with decreasing social utility, while \Cref{fig:f7,fig:f8} exhibit comparison of social utility across different mechanisms.

 \subsection{Remote city opportunity factor (RCOF) and its impact}
The flights to metro cities and financial hubs have a high valuation since the ticket prices and load factor is generally high, which gives them an advantage if the allocation is done solely based on valuation.  
The ticket fare and demand on remote routes is comparatively less due to their low purchasing power or due to cap imposed by government intervention. Hence, to provide equal opportunities to remote cities, we used the remote city opportunity factor (RCOF) given by \Cref{eq:rcof} in the slot allocation process. It assigns a higher value to cities with a low SPI and a considerable amount of population to have a benefit of air connectivity. Since the RCOF is multiplied with the valuations of the movements in the objective function of \Cref{eq:OPT1}, the chances of allocation of the remote cities in larger airports increase and at the same time, the factor $1 /\rho_i$ in the payment function leads to a lower payment for such flights.
 
We observe that the individual utilities of the flights to and from remote cities are comparable to the metro cities despite their low valuations (see \Cref{fig:f5,fig:f6}), both for the case of Delhi and Chennai airport. This happens because the average slot payment for metro cities is higher than non-metro state capitals, which in turn is higher than in remote cities. Hence, \mechabbrv\ is more egalitarian among the flights irrespective of the origin/destination.

From the regulator's point of view, $\mechabbrv$ provide a mechanism that offers opportunities to flights from remote cities and provide an incentive to airlines to operate in public service offering routes, which otherwise are neglected without regulatory intervention. On the other hand, between two cities with the same value of $\rho_i$ the city with high valuation will win the slot. 
Moreover, our mechanism is not totally unfavourable to metro cities as in our RCOF calculation, we have also incorporated the population of the cities. The metro cities with a higher population will have more slots. The policymaker can adjust the relative weight of SPI and population to decide how much importance they want to give for remote connectivity. 
Moreover, when a city makes progress on the economic front, its SPI value improves, leading to similar values of $\rho_i$ for most cities. This reduces the difference given by $\gamma_{max} - \gamma_{i}$ ($SPI_{max} - SPI_{i}$), thus reducing the preferential treatment of the city.  
In the long run, with the development of cities, the social progress gap improves and administrative approach may not be required. Therefore, the $\mechabbrv$ will evolve over time and moves in the direction of purely market-based mechanism.

As shown in \Cref{fig:f3} and \Cref{fig:f4}, an increase in congestion cost lead to more average payment by airlines. Higher payments essentially means more internalization of the congestion cost by the airlines, which earlier they imposed on society.  
This has positive benefits to society in terms of lower environmental and noise pollution along with less waiting time for the passengers. As shown, social utility decreases with an increase in congestion cost. The reason for the same is that we have only considered airline utilities in social utility calculations. Since airlines pay more by internalizing the cost of congestion, their utilities go down. However, with reduced congestion, airlines will be saving in terms of more utilization of aircrafts, crew members and less fuel burn. 

The social utility generated using $\mechabbrv$ is 20-30\% higher than IATA and Current allocation as shown in \Cref{table:1}. It is interesting to note that IATA allocation performs better then existing allocation by 10-15\%, while in turn our mechanism outperforms IATA. The properties of \mechabbrv\ help us in considering congestion cost and RCOF in the objective function and incentivize a truthful value revelation by the airlines. The payment mechanism also takes into account congestion cost and RCOF, thereby driving airline to operate in less congested slots and connect to remote cities. To the best of our knowledge, this is the first study to consider all the three goals and provide a truthful mechanism for airport slot allocation.
\begin{table}[!ht]
\caption{Percent improvement in social utility of \mechabbrv\ vis-a-vis Current and IATA based allocations.}
\centering
\begin{tabular}{ |p{1.2cm}|p{1.5cm}|p{1.8cm}|p{1.5cm}|p{1.8cm}|p{1.5cm}| } 
\hline
& & \multicolumn{2}{c|}{Chennai (MAA) Airport} &\multicolumn{2}{c|}{Delhi (DEL) Airport }\\
\hline
& & \multicolumn{4}{c|}{\% Improvement w.r.t.}\\
\hline
Time Interval & Congestion Cost &  Current allocation & IATA allocation & Current allocation & IATA allocation\\
\hline
\multirow{1}{1.2cm}{6AM - 12PM}  &	15000&	29.9\% &	4.6\%&	35.1\%&	6.2\% \\
   & 20000 &	28.4\%&	4.8\% &33.5\% &6.4\%\\
  & 25000&	26.6\%	&5.1\% &31.7\%	&6.4\%\\
  & 30000&	24.3\%	&5.6\% &30.4\%	&6.9\%\\
\hline
\multirow{1}{1.2cm}{12PM - 6PM}  & 15000&	29.9\%&	12.3\%&	43.2\%	&10.2\%\\
  &	20000&	29.1\%&	12.6\%&	42.4\%&	10.3\%\\
	& 25000&	28.4\%&	13.2\%&	41.1\%&	10.5\%\\
	& 30000&	27.0\%&	13.4\%&	39.6\%&	10.8\%\\
\hline
\multirow{1}{1.2cm}{6PM - 12AM} &	15000&	36.9\%&	11.4\%&	34.2\%&	7.1\%\\
	& 20000&	37.8\%&	11.5\%&	33.3\%&	7.3\%\\
&	25000&	38.9\%&	12.0\%&	31.9\%&	7.6\%\\
&	30000&	39.9\%&	12.8\%&	29.9\%&	8.0\%\\
\hline
\multirow{1}{1.2cm}{12AM - 6AM} &	15000&	28.7\%&	19.2\%&	32.1\%&	8.4\%\\
&	20000&	29.7\%	&19.6\%&	30.5\%&	8.7\%\\
&	25000&	31.0\%&	20.6\%&	28.5\%&	8.9\%\\
&	30000&	32.4\%&	21.6\%&	26.5\%&	9.2\%\\
\hline
\end{tabular}
\label{table:1}
\end{table}
\section{Conclusions}
\label{sec:concl}
The present study has examined the slot allocation process from a multi-objective viewpoint. We start with consideration of two major problems of the airline industry, namely congestion and regional connectivity. We combined the features of market-based and administrative instruments in a single dominant strategy incentive compatible mechanism with the goal of efficiency and social connectivity. The proposed mechanism provides a flexible and transparent approach for slot allocation, and the analytical results offer meaningful policy implications. The mechanism is solvable in polynomial time. The significant contributions of this study are as follows. 
\par First, the proposed mechanism provides a fine balance between three competing goals: efficiency, remote city connection, and congestion mitigation. A new model is proposed for optimizing the allocation of the slot based on the valuation of the movements and the congestion cost. The valuation of different movements is based on ticket prices and load factor. The model is efficient as it allocates slots to movements based on valuations. It also considers regional connectivity by incorporating the remote city opportunity factor. The weights of the population and social progress index of the cities can be flexibly changed based on the varying goals of the policymakers. In addition, the relationship between airport capacity and congestion is examined using variable capacity. It captured the phenomena of congestion when the capacity used exceeds its limits. The ‘‘threshold limit” of airport capacity can be adjusted for different airports as it depends on factors such as types of aircraft and technology used at airports. The “threshold limit” could serve as an upper limit for the non-congested capacity of the airport. Our mechanism brings in flexibility where:
\begin{itemize}
    \item Congestion and its associated cost can be tuned based on the historical congestion data of the airport.
    \item The “threshold limit” of capacity can be adjusted for different airports.
    \item The mechanism can be tweaked according to the objective of airport/priority of the policymaker, as it provides adjustable weights of social progress index and population of the remote city, which caters to social obligations for slot allocation.
    \end{itemize}
The second contribution of this study comes from the methodology used in solving the proposed formulation. Solving the objective function for maximizing social welfare using the affine maximizer is an efficient and effective approach of mechanism design. The affine maximizer mechanisms are strategy-proof and individually rational (the agents’ valuations for the chosen allocations are nonnegative). Besides, our mechanism is incentive compatible in dominant strategies. None of the players can get better utility by misreporting their true valuation for the slots. The formulation is polynomial solvable as LP relaxation yields an optimal solution in integers in polynomial time.

The third contribution of this study is its payment rule. The payment rule is designed in such a way that it captures the externality imposed by a movement on the others. It is calculated as the difference between the social welfare generated without existence of the movement and the sum of weighted valuations of all the other airlines if it participates in the allocation. The resulting payment captures two properties:
\begin{itemize}
    \item The contribution of particular movement to congestion and
    \item The type of connectivity provided by it (Metro city, non-metro capital city, or small city). 
\end{itemize}
The mechanism is designed in such a way that it assigns low payment for remote city connections and provides an incentive for airlines to operate on such routes. Moreover, the high payment value for trunk routes discourages airlines from increasing frequency only on such routes. The mechanism also considers the congestion externality imposed by a movement and charge fee based on the contribution of movement. Hence, the payment mechanism emulates the properties of Pigouvian tax, which is used to correct inefficient market by imposing a fee equal to the marginal cost of the negative externality. The proposed payment also drives our mechanism to be incentive-compatible and charge a higher cost to movements with negative externalities.

\subsection*{Limitations and Future Research Directions}
Like any other research, our study also has several limitations. Future studies can extend the body of literature on slot allocation by incorporating these limitations into their models. Firstly, our mechanism is not budget balanced. Future studies may try to achieve together both efficiency and budget balance criteria. We have considered a single airport for slot allocation.  Future studies may consider multiple airports for slot allocations simultaneously. Congestion in a single slot has a cascading effect on subsequent movements. Future studies may consider rolling capacity constraints for different slots. In our remote city connectivity factor, we have considered the social progress index and population of the cities. Future studies may also incorporate the availability of other modes of transportation such as rail, road, and water to these cities.

% Bibliography
\bibliographystyle{plainnat}
\bibliography{abb,swaprava,master,ultimate}

\begin{thebibliography}{47}
\providecommand{\natexlab}[1]{#1}
\providecommand{\url}[1]{\texttt{#1}}
\expandafter\ifx\csname urlstyle\endcsname\relax
  \providecommand{\doi}[1]{doi: #1}\else
  \providecommand{\doi}{doi: \begingroup \urlstyle{rm}\Url}\fi

\bibitem[Abdulkadiro{\u{g}}lu and S{\"o}nmez(2003)]{abdulkadirouglu2003school}
Atila Abdulkadiro{\u{g}}lu and Tayfun S{\"o}nmez.
\newblock School choice: A mechanism design approach.
\newblock \emph{American economic review}, 93\penalty0 (3):\penalty0 729--747,
  2003.

\bibitem[Alderighi et~al.(2017)Alderighi, Gaggero, and
  Piga]{alderighi2017hidden}
Marco Alderighi, Alberto~A Gaggero, and Claudio~A Piga.
\newblock The hidden side of dynamic pricing: Evidence from the airline market.
\newblock \emph{Available at SSRN 3085585}, 2017.

\bibitem[Androutsopoulos et~al.(2020)Androutsopoulos, Manousakis, and
  Madas]{androutsopoulos2020modeling}
Konstantinos~N Androutsopoulos, Eleftherios~G Manousakis, and Michael~A Madas.
\newblock Modeling and solving a bi-objective airport slot scheduling problem.
\newblock \emph{European Journal of Operational Research}, 284\penalty0
  (1):\penalty0 135--151, 2020.

\bibitem[Anstee(1987)]{anstee1987polynomial}
Richard~P Anstee.
\newblock A polynomial algorithm for b-matchings: an alternative approach.
\newblock \emph{Information Processing Letters}, 24\penalty0 (3):\penalty0
  153--157, 1987.

\bibitem[Ball et~al.(2006)Ball, Donohue, and Hoffman]{ball2006auctions}
Michael Ball, George Donohue, and Karla Hoffman.
\newblock Auctions for the safe, efficient, and equitable allocation of
  airspace system resources.
\newblock \emph{Combinatorial auctions}, 1, 2006.

\bibitem[Ball et~al.(2020)Ball, Estes, Hansen, and Liu]{ball2020quantity}
Michael~O Ball, Alexander~S Estes, Mark Hansen, and Yulin Liu.
\newblock Quantity-contingent auctions and allocation of airport slots.
\newblock \emph{Transportation Science}, 54\penalty0 (4):\penalty0 858--881,
  2020.

\bibitem[Barnhart and Cohn(2004)]{barnhart2004airline}
Cynthia Barnhart and Amy Cohn.
\newblock Airline schedule planning: Accomplishments and opportunities.
\newblock \emph{Manufacturing \& service operations management}, 6\penalty0
  (1):\penalty0 3--22, 2004.

\bibitem[Basso and Zhang(2010)]{basso2010pricing}
Leonardo~J Basso and Anming Zhang.
\newblock Pricing vs. slot policies when airport profits matter.
\newblock \emph{Transportation Research Part B: Methodological}, 44\penalty0
  (3):\penalty0 381--391, 2010.

\bibitem[Benell and Prentice(1993)]{benell1993regression}
Dave~W Benell and Barry~E Prentice.
\newblock A regression model for predicting the economic impacts of canadian
  airports.
\newblock \emph{Logistics and Transportation Review}, 29\penalty0 (2):\penalty0
  139, 1993.

\bibitem[Bilotkach et~al.(2015)Bilotkach, Gaggero, and
  Piga]{bilotkach2015airline}
Volodymyr Bilotkach, Alberto~A Gaggero, and Claudio~A Piga.
\newblock Airline pricing under different market conditions: Evidence from
  european low-cost carriers.
\newblock \emph{Tourism Management}, 47:\penalty0 152--163, 2015.

\bibitem[Brafman and Tennenholtz(2003)]{Brafman2003}
R.I. Brafman and M.~Tennenholtz.
\newblock R-max-a general polynomial time algorithm for near-optimal
  reinforcement learning.
\newblock \emph{The Journal of Machine Learning Research}, 3:\penalty0
  213--231, 2003.

\bibitem[Brueckner(2003)]{brueckner2003airline}
Jan~K Brueckner.
\newblock Airline traffic and urban economic development.
\newblock \emph{Urban Studies}, 40\penalty0 (8):\penalty0 1455--1469, 2003.

\bibitem[Brueckner(2009)]{brueckner2009price}
Jan~K Brueckner.
\newblock Price vs. quantity-based approaches to airport congestion management.
\newblock \emph{Journal of Public Economics}, 93\penalty0 (5-6):\penalty0
  681--690, 2009.

\bibitem[Burghouwt(2017)]{burghouwt2017influencing}
Guillaume Burghouwt.
\newblock Influencing air connectivity outcomes.
\newblock International Transport Forum Discussion Paper, 2017.

\bibitem[Cao and Kanafani(2000)]{cao2000value}
Jia-Ming Cao and Adib Kanafani.
\newblock The value of runway time slots for airlines.
\newblock \emph{European Journal of Operational Research}, 126\penalty0
  (3):\penalty0 491--500, 2000.

\bibitem[Carlin and Park(1970)]{carlin1970marginal}
Alan Carlin and Rolla~Edward Park.
\newblock Marginal cost pricing of airport runway capacity.
\newblock \emph{The American Economic Review}, pages 310--319, 1970.

\bibitem[Castelli et~al.(2012)Castelli, Pellegrini, and
  Pesenti]{castelli2012airport}
Lorenzo Castelli, Paola Pellegrini, and Raffaele Pesenti.
\newblock Airport slot allocation in europe: economic efficiency and fairness.
\newblock \emph{International journal of revenue management}, 6\penalty0
  (1-2):\penalty0 28--44, 2012.

\bibitem[Czerny(2010)]{czerny2010airport}
Achim~I Czerny.
\newblock Airport congestion management under uncertainty.
\newblock \emph{Transportation Research Part B: Methodological}, 44\penalty0
  (3):\penalty0 371--380, 2010.

\bibitem[Czerny and Zhang(2014)]{czerny2014airport}
Achim~I Czerny and Anming Zhang.
\newblock Airport congestion pricing when airlines price discriminate.
\newblock \emph{Transportation Research Part B: Methodological}, 65:\penalty0
  77--89, 2014.

\bibitem[Daniel(2011)]{daniel2011congestion}
Joseph~I Daniel.
\newblock Congestion pricing of canadian airports.
\newblock \emph{Canadian Journal of Economics/Revue canadienne
  d'{\'e}conomique}, 44\penalty0 (1):\penalty0 290--324, 2011.

\bibitem[Daniel and Harback(2009)]{daniel2009pricing}
Joseph~I Daniel and Katherine~Thomas Harback.
\newblock Pricing the major us hub airports.
\newblock \emph{Journal of Urban Economics}, 66\penalty0 (1):\penalty0 33--56,
  2009.

\bibitem[de~Werra(1981)]{de1981some}
Dominique de~Werra.
\newblock On some characterisations of totally unimodular matrices.
\newblock \emph{Mathematical Programming}, 20\penalty0 (1):\penalty0 14--21,
  1981.

\bibitem[Deshpande and Ar{\i}kan(2012)]{deshpande2012impact}
Vinayak Deshpande and Mazhar Ar{\i}kan.
\newblock The impact of airline flight schedules on flight delays.
\newblock \emph{Manufacturing \& Service Operations Management}, 14\penalty0
  (3):\penalty0 423--440, 2012.

\bibitem[Fageda et~al.(2018)Fageda, Su{\'a}rez-Alem{\'a}n, Serebrisky, and
  Fioravanti]{fageda2018air}
Xavier Fageda, Ancor Su{\'a}rez-Alem{\'a}n, Tomas Serebrisky, and Reinaldo
  Fioravanti.
\newblock Air connectivity in remote regions: A comprehensive review of
  existing transport policies worldwide.
\newblock \emph{Journal of Air Transport Management}, 66:\penalty0 65--75,
  2018.

\bibitem[Fan and Odoni(2002)]{fan2002practical}
Terence~P Fan and Amedeo~R Odoni.
\newblock A practical perspective on airport demand management.
\newblock \emph{Air Traffic Control Quarterly}, 10\penalty0 (3):\penalty0
  285--306, 2002.

\bibitem[Ghouila-Houri(1962)]{ghouila1962caracterisation}
Alain Ghouila-Houri.
\newblock Caract{\'e}risation des matrices totalement unimodulaires.
\newblock \emph{Comptes Redus Hebdomadaires des S{\'e}ances de l'Acad{\'e}mie
  des Sciences (Paris)}, 254:\penalty0 1192--1194, 1962.

\bibitem[Green(2007)]{green2007airports}
Richard~K Green.
\newblock Airports and economic development.
\newblock \emph{Real estate economics}, 35\penalty0 (1):\penalty0 91--112,
  2007.

\bibitem[Gr{\"o}tschel et~al.(1993)Gr{\"o}tschel, Lov{\'a}sz, and
  Schrijver]{grotschel1993complexity}
Martin Gr{\"o}tschel, L{\'a}szl{\'o} Lov{\'a}sz, and Alexander Schrijver.
\newblock Complexity, oracles, and numerical computation.
\newblock In \emph{Geometric Algorithms and Combinatorial Optimization}, pages
  21--45. Springer, 1993.

\bibitem[Harsha(2009)]{harsha2009mitigating}
Pavithra Harsha.
\newblock \emph{Mitigating airport congestion: market mechanisms and airline
  response models}.
\newblock PhD thesis, Massachusetts Institute of Technology, 2009.

\bibitem[Horton(2020)]{horton_2020}
Will Horton.
\newblock British airways cuts threaten crown jewel of slots at london heathrow
  and gatwick.
\newblock \emph{The Forbes}, May 2020.
\newblock URL
  \url{https://www.forbes.com/sites/willhorton1/2020/05/04/british-airways-cuts-threaten-crown-jewel-of-slots-at-london-heathrow-and-gatwick/#359a052b7cdf}.

\bibitem[IATA(2017)]{IATA}
IATA.
\newblock Worldwide slot guidelines, eighth ed.
\newblock \emph{IATA}, 2017.
\newblock URL
  \url{https://www.iata.org/policy/slots/Documents/wsg-8-english.pdf}.

\bibitem[Li et~al.(2010)Li, Lam, Wong, and Fu]{li2010optimal}
Zhi-Chun Li, William~HK Lam, SC~Wong, and Xiaowen Fu.
\newblock Optimal route allocation in a liberalizing airline market.
\newblock \emph{Transportation Research Part B: Methodological}, 44\penalty0
  (7):\penalty0 886--902, 2010.

\bibitem[Mehta and Vazirani(2020)]{mehta2020incentive}
Ruta Mehta and Vijay~V Vazirani.
\newblock An incentive compatible, efficient market for air traffic flow
  management.
\newblock \emph{Theoretical Computer Science}, 818:\penalty0 41--50, 2020.

\bibitem[Ribeiro et~al.(2018)Ribeiro, Jacquillat, Antunes, Odoni, and
  Pita]{ribeiro2018optimization}
Nuno~Antunes Ribeiro, Alexandre Jacquillat, Ant{\'o}nio~Pais Antunes, Amedeo~R
  Odoni, and Jo{\~a}o~P Pita.
\newblock An optimization approach for airport slot allocation under iata
  guidelines.
\newblock \emph{Transportation Research Part B: Methodological}, 112:\penalty0
  132--156, 2018.

\bibitem[Roberts(1979)]{roberts1979characterization}
Kevin Roberts.
\newblock The characterization of implementable choice rules.
\newblock \emph{Aggregation and revelation of preferences}, 12\penalty0
  (2):\penalty0 321--348, 1979.

\bibitem[Sheard(2014)]{sheard2014airports}
Nicholas Sheard.
\newblock Airports and urban sectoral employment.
\newblock \emph{Journal of Urban Economics}, 80:\penalty0 133--152, 2014.

\bibitem[Sheng et~al.(2019)Sheng, Li, and Fu]{sheng2019modeling}
Dian Sheng, Zhi-Chun Li, and Xiaowen Fu.
\newblock Modeling the effects of airline slot hoarding behavior under the
  grandfather rights with use-it-or-lose-it rule.
\newblock \emph{Transportation Research Part E: Logistics and Transportation
  Review}, 122:\penalty0 48--61, 2019.

\bibitem[Shoham and Leyton-Brown(2008)]{SL08}
Y.~Shoham and K.~Leyton-Brown.
\newblock \emph{Multiagent Systems: Algorithmic, Game-Theoretic, and Logical
  Foundations}.
\newblock Cambridge University Press, 2008.

\bibitem[Sieg(2010)]{sieg2010grandfather}
Gernot Sieg.
\newblock Grandfather rights in the market for airport slots.
\newblock \emph{Transportation Research Part B: Methodological}, 44\penalty0
  (1):\penalty0 29--37, 2010.

\bibitem[Swaroop et~al.(2012)Swaroop, Zou, Ball, and Hansen]{swaroop2012more}
Prem Swaroop, Bo~Zou, Michael~O Ball, and Mark Hansen.
\newblock Do more us airports need slot controls? a welfare based approach to
  determine slot levels.
\newblock \emph{Transportation Research Part B: Methodological}, 46\penalty0
  (9):\penalty0 1239--1259, 2012.

\bibitem[Tamir and Mitchell(1998)]{tamir1998maximumb}
Arie Tamir and Joseph~SB Mitchell.
\newblock A maximumb-matching problem arising from median location models with
  applications to the roommates problem.
\newblock \emph{Mathematical Programming}, 80\penalty0 (2):\penalty0 171--194,
  1998.

\bibitem[Vaze and Barnhart(2012)]{vaze2012modeling}
Vikrant Vaze and Cynthia Barnhart.
\newblock Modeling airline frequency competition for airport congestion
  mitigation.
\newblock \emph{Transportation Science}, 46\penalty0 (4):\penalty0 512--535,
  2012.

\bibitem[Vickrey(1969)]{vickrey1969congestion}
William~S Vickrey.
\newblock Congestion theory and transport investment.
\newblock \emph{The American Economic Review}, 59\penalty0 (2):\penalty0
  251--260, 1969.

\bibitem[Yao and Yang(2008)]{yao2008airport}
Shujie Yao and Xiuyun Yang.
\newblock Airport development and regional economic growth in china.
\newblock \emph{Available at SSRN 1101574}, 2008.

\bibitem[Zhang et~al.(2020)Zhang, Atasu, Ayer, and Toktay]{zhang2020truthful}
Can Zhang, Atalay Atasu, Turgay Ayer, and L~Beril Toktay.
\newblock Truthful mechanisms for medical surplus product allocation.
\newblock \emph{Manufacturing \& Service Operations Management}, 22\penalty0
  (4):\penalty0 735--753, 2020.

\bibitem[Zografos and Madas(2003)]{zografos2003critical}
Konstantinos~G Zografos and Michael~A Madas.
\newblock Critical assessment of airport demand management strategies in europe
  and the united states: Comparative perspective.
\newblock \emph{Transportation research record}, 1850\penalty0 (1):\penalty0
  41--48, 2003.

\bibitem[Zografos et~al.(2013)Zografos, Madas, and
  Salouras]{zografos2013decision}
Konstantinos~G Zografos, Michael~A Madas, and Yiannis Salouras.
\newblock A decision support system for total airport operations management and
  planning.
\newblock \emph{Journal of Advanced Transportation}, 47\penalty0 (2):\penalty0
  170--189, 2013.

\end{thebibliography}

\appendix
\section*{Appendix}

\begin{proof}[Proof of \Cref{thm:DSIC}]
  In terms of utilities, the theorem says,
\begin{equation}
    u_i(f(v_i,v_{-i}),V ) \geq u_i(f(v'_{i},v_{-i}),V ),\ \ \  \forall v_i, v'_{i} \in \mathbb{R}^{|n|} \ \ \forall i \in M
\end{equation}
Let us assume for the contradiction that, there exist a movement $i$ for which the corresponding airline has true valuations for the slots as, $v_{i}$, but misreports it as $v'_{i}$(the corresponding value function is $v'_{i}$), and gets better utility,
\begin{equation}
    u_i(f(v'_{i},v_{-i}),V) >
u_i(f(v_i,v_{-i}),V)
\label{eq:Eq2DSIC}
\end{equation}
Suppose $\mathcal{A}(v'_{i},v_{-i}) = x'$ and $\mathcal{A}({v}_{i},v_{-i}) = x^{*}$. By definition of the utility function, 
\begin{equation*}
    u_i( x^{*},V) = v_i(x^{*})- p_i({v}_{i},v_{-i})
\end{equation*}
\begin{equation*}
    = v_i(x^{*})- \frac{1}{\rho_i } \Big( h_i(v_{-i}) - \Big(\sum\limits_{k \in M \setminus \{i\}} \rho_k\ v_k\big(\ x^{*} \ \big) - g\sum\limits_{j \in S}\ e_j(x^{*}) \Big)  \Big)
\end{equation*}
\begin{equation*}
    = v_i(x^{*})+ \frac{1}{\rho_i } \Big(\sum\limits_{k \in M \setminus \{i\}} \rho_k\ v_k\big(\ x^{*} \ \big) - g\sum\limits_{j \in S}\ e_j(x^{*}) \Big)
    - \frac{1}{\rho_i } \Big( h_i(v_{-i}) \Big)
\end{equation*}
\begin{equation}
    =  \frac{1}{\rho_i } \Big(\sum\limits_{k \in M } \rho_k\ v_k\big(\ x^{*} \ \big) - g\sum\limits_{j \in S}\ e_j(x^{*}) \Big)
    - \frac{1}{\rho_i } \Big( h_i(v_{-i}) \Big)
\end{equation}
and,
\begin{equation*}
    u_i( x',({v}_{i},v_{-i})) = v_i(x')- p_i(v'_{i},v_{-i})
\end{equation*}
\begin{equation}
    =  \frac{1}{\rho_i } \Big(\sum\limits_{k \in M } \rho_k\ v_k\big(\ x' \ \big) - g\sum\limits_{j \in S}\ e_j(x') \Big)
    - \frac{1}{\rho_i } \Big( h_i(v_{-i}) \Big)
\end{equation}
From inequality \ref{eq:Eq2DSIC}, 
\begin{align*}
    \frac{1}{\rho_i } \Big(\sum\limits_{k \in M } \rho_k\ v_k\big(\ x' \ \big) - g\sum\limits_{j \in S}\ e_j(x') \Big)
    - & \frac{1}{\rho_i } \Big( h_i(v_{-i}) \Big) \\ \ >\  \frac{1}{\rho_i } \Big(\sum\limits_{k \in M } \rho_k\ v_k\big(\ x^{*} \ \big) - & g\sum\limits_{j \in S}\ e_j(x^{*}) \Big)
    - \frac{1}{\rho_i } \Big( h_i(v_{-i}) \Big)
\end{align*}
\begin{equation*}
     \sum\limits_{k \in M } \rho_k\ v_k\big(\ x' \ \big) - g\sum\limits_{j \in S}\ e_j(x') 
     \ >\  \sum\limits_{k \in M } \rho_k\ v_k\big(\ x^{*} \ \big) - g\sum\limits_{j \in S}\ e_j(x^{*}) 
\end{equation*}
The above inequality leads to the contradiction with the fact that $x^{*}$ is the socially optimal allocation. And therefore, the best strategy for airlines is to report the valuations of their movements truthfully.
\end{proof}

\begin{proof}[Proof of \Cref{thm:IR}]
 Consider the allocation given by $\mathcal{A}$ as $x^*$. The utility of the airline $i$ is, $ u_i((x^*,p(V)),V)= v_i (x^*)- p_i(V)$ $\forall i \in M$. $v_i (\mathcal{A}(V)) \geq 0$, as it is the valuation by the allocation of slots. If $\rho_i =0$ then $u_i((x^*,p(V)),V)\geq 0$. For the case when $\rho_i >0$, the proof is as follows:  by expanding the expression for $p$ in  the definition of utility, we get
 \begin{equation*}
    u_i((x^*,p(V)),V)= v_i (x^*)- \frac{1}{\rho_i } \Big( h_i(v_{-i}) - \Big(\sum\limits_{k \in M \setminus \{i\}} \rho_k\ v_k(x^*) -g\sum\limits_{j \in S}e_j(x^*) \Big)  \Big)
\end{equation*}
\begin{equation*}
   = v_i (x^*)- \frac{1}{\rho_i } \Big( h_i(v_{-i})\Big) + \frac{1}{\rho_i }\Big( \sum\limits_{k \in M \setminus \{i\}} \rho_k\ v_k(\ x^* \ ) - g\sum\limits_{j \in S}\ e_j(\ x^*\ ) \Big) 
\end{equation*}
\begin{align*}
   = \frac{1}{\rho_i }\Big( \sum\limits_{k \in M } \rho_k\ v_k(\ x^* \ ) - g\sum\limits_{j \in S}\ e_j(\ x^*\ ) \Big) & \\ - \frac{1}{\rho_i } \Big( \sum\limits_{k \in M \setminus \{i\}} \rho_k\ v_k\big(\mathcal{A}(v_{-i}) \big) & - g\sum\limits_{j \in S}\ e_j(\mathcal{A}(v_{-i})\ )\Big)
\end{align*}
By adding and subtracting $\frac{1}{\rho_i }(v_i(\mathcal{A}(v_{-i})))$,
\begin{align*}
   = \frac{1}{\rho_i }\Big( \sum\limits_{k \in M } \rho_k\ v_k(\ x^* \ ) - g\sum\limits_{j \in S}\ e_j(\ x^*\ ) \Big) & \\ - \frac{1}{\rho_i } \Big( \sum\limits_{k \in M } \rho_k\ v_k\big(\mathcal{A}(v_{-i}) \big) & - g\sum\limits_{j \in S}\ e_j(\mathcal{A}(v_{-i}) )\Big)+ \frac{1}{\rho_i }(v_i(\mathcal{A}(v_{-i}))) 
\end{align*}

\begin{align*}
   = \underbrace{ \Big( \sum\limits_{k \in M } \rho_k\ v_k(x^*) - g\sum\limits_{j \in S}e_j(x^*) \Big) - \Big( \sum\limits_{k \in M } \rho_k\ v_k\big(\mathcal{A}(v_{-i}) \big) - g\sum\limits_{j \in S}e_j(\mathcal{A}(v_{-i}) )\Big)}_{\geq 0} & \\ + \underbrace{v_i(\mathcal{A}(v_{-i})) }_{\geq 0} &
\end{align*}
The difference of first two terms is non-negative by the definition of $x^*$, the socially optimal allocation, and the third term is non-negative as it is the valuation by the allocation of slots, which proves that $\mechabbrv$ is individually rational.
\end{proof}

\end{document}